\documentclass[smallextended]{svjour3}
\usepackage{amsmath,graphicx,indentfirst,enumitem,amssymb,
tikz,color,geometry,url}
\usepackage[subsection]{algorithm}
\usepackage{color}
\usepackage{amsfonts}
\usepackage{algorithm}
\usepackage{algpseudocode}
\usetikzlibrary{positioning,calc,arrows,automata}
\newcommand{\R}{\mathbb{R}}

\date{}
\begin{document}
\author{Silvia Noschese \and Lothar Reichel}
\institute{Silvia Noschese\at
Dipartimento di Matematica\\ 
SAPIENZA Universit\`a di Roma\\
P.le Aldo Moro 5, 00185 Roma, Italy\\
\email{noschese@mat.uniroma1.it \\
{\it Corresponding author.}}\\[2mm]
Lothar Reichel\at
Department of Mathematical Sciences\\ 
Kent State University\\
Kent, OH 44242, USA\\
\email{reichel@math.kent.edu}
}
\title{Network connectivity analysis via shortest paths}
\maketitle

\begin{abstract}
 {Complex systems of interacting components often can be modeled by a simple graph 
$\mathcal{G}$ that consists of a set of $n$ nodes and a set of $m$ edges. Such a graph can
be represented by an adjacency matrix $A\in\R^{n\times n}$, whose $(ij)$th entry is one if 
there is an edge pointing from node $i$ to node $j$, and is zero otherwise. The matrix 
$A$ and its positive integer powers reveal important properties of the graph and allow
the construction of the path length matrix $L$ for the graph. The $(ij)$th entry of $L$
is the length of the shortest path from node $i$ to node $j$; if there is no path between 
these nodes, then the value of the entry is set to $\infty$. We are interested in how well 
information flows via shortest paths of the graph. This can be studied with the aid of the
path length matrix. The path length matrix allows the definition of several measures of
communication in the network defined by the graph such as the global $K$-efficiency, which
considers shortest paths that are made up of at most $K$ edges for some $K<n$, as well as 
the number of such shortest paths. Novel notions of connectivity introduced in this paper 
help us understand the importance of specific edges for the flow of information through 
the graph. This is of interest when seeking to simplify a network by removing selected 
edges or trying to assess the sensitivity of the flow of information to changes due to 
exterior causes such as a traffic stoppage on a road network.}
\end{abstract}


\section{Introduction}\label{s1}
An important characteristic of a network is the ease with which one part of the network 
can be reached from another part by following its edges. Borgatti \cite{Bor} introduced
an interesting classification of centrality measures based on how traffic flows through 
a network. We are especially interested in how well information flows through a network 
via shortest paths. Such measures, e.g., the diameter of the graph that represents the 
network or the global efficiency of the network, aim to shed light on patterns of 
communication in a network as well as on the effect of edge removals. 

Let us introduce some notation and definitions that will be used throughout this paper. A 
network is represented by a graph
$\mathcal{G}=\langle\mathcal{V},\mathcal{E}\rangle$, which consists of a set of 
\emph{nodes} or \emph{vertices} $\mathcal{V}=\{1,2,\dots,n\}$, and a set of \emph{edges} 
$\mathcal{E}=\{e_1,e_2,\dots,e_m\}$ that connect the nodes. A graph is said to be 
\emph{unweighted} if the edges are treated as binary interactions, i.e., there are no 
associated weights. An edge is referred to as \emph{directed} if it starts at a node $i$ 
and ends at a node $j$, and is denoted by $e(i\rightarrow j)$. An edge between the nodes 
$i$ and $j$ is said to be \emph{undirected} when the pair of nodes is unordered, and
is denoted by $e(i\leftrightarrow j)$. A graph with only undirected edges is said to be 
\emph{undirected}; otherwise the graph is \emph{directed}. We refer to an unweighted,
possibly directed, graph without multiple edges or self-loops as a \emph{simple graph}; 
see, e.g., \cite[Chapter 6]{newman2010}. This work considers simple graphs. We remark that
the graphs discussed in this paper may have disconnected components.

The \emph{adjacency matrix} $A=[a_{ij}]_{i,j=1}^n\in\R^{n\times n}$ for a  graph 
$\mathcal{G}$ has the entry $a_{ij}=1$ if there is an edge $e(i\rightarrow j)$ in 
$\mathcal{G}$. If there is no edge $e(i\rightarrow j)$ in $\mathcal{G}$, then $a_{ij}=0$. 
In particular, the diagonal entries of the adjacency matrix for a simple graph vanish. 
Typically, the adjacency matrix of a graph is sparse, i.e., it has many more zero entries
than entries equal to one.

The \emph{out-degree} of node $i$ is defined as the number of edges that start at node 
$i$, i.e., 
\[
{\rm deg}_i^{\rm out}=\sum_{k=1}^n a_{ik},
\]
and the \emph{in-degree} of node $i$ is the number of edges that end at node $i$, i.e.,
\[
{\rm deg}_i^{\rm in}=\sum_{k=1}^n a_{ki}.
\]

A sequence of $k$ edges (not necessarily distinct) such that 
$\{e(j_1\rightarrow j_2),e(j_2\rightarrow j_3),\ldots,e(j_k\rightarrow j_{k+1})\}$ form a
\emph{walk} of length $k$. If $j_{k+1}=j_1$, then the walk is said to be \emph{closed}. If
the  {nodes} in a walk are distinct, then the walk is referred to as a \emph{path}. Thus, the 
maximal length of a path is $n-1$. A simple (directed) graph $\mathcal{G}$ is said to be 
(strongly) {\em connected} if every node is reachable by a path from every other node. 
This property is equivalent to the irreducibility of the adjacency matrix associated with 
$\mathcal{G}$. For further discussions on the above definitions and networks, as well as 
for many applications; see \cite{estrada2011structure,newman2010}.

The entry $a^{(k)}_{ij}$ of the matrix $A^k=[a^{(k)}_{ij}]_{i,j=1}^n\in\R^{n\times n}$ 
counts the number of walks of length $k$ from node $i$ to node $j$. A matrix function 
based on the powers $A^k$, $k=0,1,\ldots~$, that is analytic at the origin and vanishes 
there can be defined by a Maclaurin series
\begin{equation}\label{matfun}
f(A)=\sum_{k=1}^{\infty}c_kA^k,
\end{equation}
where the choice of coefficients $c_k$ is such that the series converges. Matrix functions
applied in network analysis generally have the property that $0\leq c_{k+1}\leq c_k$ for 
all $k\geq1$, because information flows more easily through short walks than through long 
ones; see, e.g., \cite{AMDLCR} for a discussion. The most common matrix function used in 
network analysis is the matrix exponential; see 
\cite{AB,BFRR,DMR,DMR2,DMR3,estrada2011structure,EH,estrada2005subgraph,FMRR,FRR} for 
discussions and illustrations. However, typically the \emph{modified matrix exponential}
$\exp_0(A):=\exp(A)-I$, where $I$ denotes the identity matrix, is preferable, because the 
first term in the Maclaurin series for $\exp(A)$ has no natural interpretation in the 
context of network modeling. For the modified matrix exponential, we have $c_k=1/k!$, and 
the series \eqref{matfun} converges for any adjacency matrix $A$. 

The \emph{communicability} between distinct nodes $i$ and $j$, $i\ne j$, is defined by
\begin{equation}\label{expcom}
[\exp_0(A)]_{ij}=\sum_{k=1}^{\infty}\frac{a^{(k)}_{ij}}{k!};
\end{equation}
see \cite{estrada2011structure,EH} for the analogous definition based on $\exp(A)$. The 
larger the value of $[\exp_0(A)]_{ij}$, the better is the communicability between nodes 
$i$ and $j$. The communicability accounts for all possible routes of communication between
nodes $i$ and $j$ in the network defined by the adjacency matrix $A$, and assigns a larger
weight to shorter walks than to longer ones.  {When $j=i$ the quantity \eqref{expcom} is
referred to as the \emph{subgraph centrality} of node $i$; see \cite{estrada2005subgraph}.}

We are interested in how well information flows via shortest paths. These paths 
provide the most important transmission of information between nodes. For instance, 
when considering transportation from a source node to a termination node, shortest path 
algorithms are used in GPS navigation apps such as Google Maps. We will describe how to 
construct suitable matrices, based on the shortest paths or on the shortest paths with at
most a specified number of edges, that shed light on geodesic communication in a network.

 {Let $K$ be a fixed integer such that $1\leq K<n$. A matrix of particular interest to 
us is the $K$-\emph{path length matrix} associated with the network. It is defined as
\[
L^{(K)}=[\ell^{(K)}_{ij}]_{i,j=1}^n\in\R^{n\times n},
\]
where $\ell^{(K)}_{ij}=q$ if there exists a shortest path from node $i$ to node $j$ that 
is made up of $q$ edges for $0\leq q\leq K$. If every path from node $i$ to node $j$ is 
made up of more than $K$ edges, or if $i \ne j$ and there is no path from $i$ to $j$, then 
$\ell^{(K)}_{ij}=\infty$. The diagonal entries of $L^{(K)}$ are set to zero.}

We also introduce the $K$-\emph{path counter matrix} 
\[
C^{(K)}=[c^{(K)}_{ij}]_{i,j=1}^n\in\R^{n\times n} 
\]
associated with the network, where $c^{(K)}_{ij}=p$ if there are $p$ shortest paths from 
node $i$ to node $j$ made up of at most $K$ edges; if every path from node $i$ to node $j$
is made up of more than $K$ edges, or if there is no path from $i$ to $j$, then 
$c^{(K)}_{ij}=0$. The diagonal entries of $C^{(K)}$ are set to one. This matrix sheds 
light on the effect of removing selected edges from the graph ${\mathcal G}$. 

This paper is organized as follows: Section \ref{s1.5} introduces notation to be used in
the sequel and Section
\ref{s2} describes properties of the $K$-path length matrix $L^{(K)}$ as well as of the 
$K$-path counter matrix $C^{(K)}$. In Section \ref{s3} we discuss properties of the global 
$K$-connectivity, which easily can be computed by means of the matrices $L^{(K)}$ and 
$C^{(K)}$. We also introduce the notion of edge connectivity centrality. Section \ref{s4} 
investigates which edges should be removed in order to decrease the global connectivity 
the most. We discuss the role played by the novel measure of edge connectivity centrality, 
which refines the well-known notion of edge betweenness centrality \cite{GN}. A few 
computed illustrations are reported in Section \ref{s5} and concluding remarks can be found
in Section \ref{s6}.

 {We conclude this section with some comments on related work. Brown and Colburn 
\cite{BC} discuss the computation of the roots of the reliability polynomial. These 
polynomials give the probability that a graph remains connected when some edges fail 
independently with a specified probability. However, the calculation of these polynomials 
generally is NP-hard. They therefore only can be computed for small graphs.}

 {A variety of ways to discuss the importance of edges are discussed in
\cite{AB,DJNR,NR,Sc}. These works use the matrix exponential and 
consider changes in the total communicability, defined as ${\bf 1}^T\exp(A)\bf{1}$, to 
identify edges whose removal has a large impact on communication. Here $A$ denotes the
adjacency matrix for the graph and ${\bf 1}$ is the vector with all entries one. 
De la Cruz Cabrera et al. \cite{DMR2,DMR3} seek to identify the 
most important nodes in the line graph associated with ${\mathcal G}$. These nodes 
correspond to the most important edges in ${\mathcal G}$. It is illustrated that this 
approach may identify more appropriate edges than the technique in \cite{AB}, but the 
computational burden is large for graphs with many edges. The method proposed in the 
present paper is much less demanding computationally. The sensitivity of the total 
communicability to changes in edge weights in a graph is discussed in \cite{NR,Sc}. All
methods mentioned assume that the communicability measure \eqref{expcom} is appropriate.}

 {We are interested in exploring communication along the shortest path(s) between 
nodes and we would like to take the number of distinct shortest paths into account. This 
requires the use of other measures of communicability, including measures that are related
to betweenness centrality. The betweenness centrality of a node $i$ is defined as the sum 
over all ordered pairs of distinct nodes $(h,k)$ of the fraction of shortest paths from 
node $h$ to node $k$ that pass through node $i$ (where $i$ is not an endpoint); see, e.g.,
\cite{estrada2011structure,Fr,newman2010}. This measure can be used to estimate the 
importance of an edge. The present paper proposes modifications of this measure.}

\section{Preliminaries}\label{s1.5}
The \emph{diameter} $d_{\mathcal{G}}$ of a connected graph $\mathcal{G}$ is the maximal
length of the shortest path between any pair of distinct nodes of the graph, i.e.,
\begin{equation} \label{d_G}
d_{\mathcal{G}}=\max_{1\leq i,j\leq n} \ell_{ij}^{(n-1)}.
\end{equation}
This quantity provides a measure of how easy it is for the nodes of a graph to 
communicate. 

We refer to the $K$-path length matrix for $K=n-1$,
\[
L=[\ell_{ij}]_{i,j=1}^n=L^{(n-1)},
\]
as the \emph{path length matrix} and to the $K$-path counter matrix for $K=n-1$,
\[
C=[c_{ij}]_{i,j=1}^n=C^{(n-1)},
\]
as the \emph{path counter matrix}. The path length matrix $L$ allows us to define the 
\emph{global efficiency} of a graph, which is the average inverse geodesic length of 
$\mathcal{G}$,
\begin{equation}\label{e_G}
e_{\mathcal{G}}=\frac{1}{n(n-1)}\sum_{i,j\ne i} \frac{1}{\ell_{ij}},
\end{equation}
where $1/\infty$ is identified with $0$; see \cite{BBV,HKYH}. We remark that in the 
context of molecular chemistry, the sum of the reciprocals of the lengths between all 
pairs of nodes of an undirected, unweighted, connected graph is known as the 
\emph{Harary index} \cite{PNTM}. The notion of global efficiency also is useful when the
network has more than one connected component, because infinite lengths do not contribute
to the sum \eqref{e_G}. Taking into account only shortest paths made up of at most $K$ 
edges, with $1\leq K<n$, leads to the \emph{global $K$-efficiency} 
\begin{equation}\label{ekg}
e^{(K)}_{\mathcal{G}}=\frac{1}{n(n-1)}\sum_{i,j\ne i}\frac{1}{\ell^{(K)}_{ij}},
\end{equation}
which was introduced in \cite{NR_na}. 

The larger the number of shortest paths between a pair of nodes is, the smaller is the risk 
of making the graph disconnected by removing an edge of such a path. The average 
number of geodesics normalized by their length can be computed with the aid
of the path counter matrix $C$:
\begin{equation}\label{c_G}
c_{\mathcal{G}}=\frac{1}{n(n-1)}\sum_{i,j\ne i} \frac{c_{ij}}{\ell_{ij}}.
\end{equation}
This measure, which was introduced by Li et al.  \cite{LPWS}, is referred to as the 
\emph{global connectivity}. We define the \emph{global $K$-connectivity}, 
\begin{equation}\label{ckG}
c^{(K)}_{\mathcal{G}}=\frac{1}{n(n-1)}\sum_{i,j\ne i}\frac{c^{(K)}_{ij}}{\ell^{(K)}_{ij}},
\end{equation}
which only takes into account the shortest paths that are made up of at most $K$ edges, 
with $1\leq K<n$.

It is of interest to analyze the effect of the removal of selected edges of $\mathcal{G}$
on these connectivity measures. We will introduce suitable edge centrality measures that
guide our choice of edges to be removed.

\section{The $K$-path length and $K$-path counter matrices}\label{s2}
Recall that the $(ij)$th entry $a^{(h)}_{ij}$ of the matrix $A^h$ counts the number of
walks of length $h$ from node $i$ to node  $j$. It therefore is straightforward to show 
the following results. They lead to the construction of the $K$-path length matrix 
$L^{(K)}$ and of the $K$-path counter matrix $C^{(K)}$ for integers $K$ with $1\leq K<n$.

\begin{proposition}\label{fnnz}
Let the graph $\cal{G}=\langle\mathcal{V},\mathcal{E}\rangle$ be connected. For all  
$i,j \in \mathcal{V}$, $i\ne j$, if  $e(i\rightarrow j)\in \mathcal{E}$, i.e., if
$a_{ij}^{(1)}=a_{ij}=1$, then there exists a shortest path of length $q=1$ from node $i$ 
to node $j$, and $\ell_{ij}^{(h)}=1$, for all  $h$, $1\leq h<n$. If, instead, 
$e(i\rightarrow j)\notin\mathcal{E}$, then let $q$, for some $1<q<n$, be the smallest 
power of $A$ such that $a_{ij}^{(q)}=p>0$, that is
\begin{equation*}
a^{(h)}_{ij}=0,\quad \forall h,~~~ 1<h<q,\quad \mathrm{and} \quad a^{(q)}_{ij}=p.
\end{equation*}
Then the length of the shortest path from node $i$ to node $j$ is $q$, i.e.,
\[
\ell_{ij}^{(h)}=\infty,\quad \forall h,~~~ 1< h<q,\quad \mathrm{and}  \quad 
\ell_{ij}^{(h)} = q,\quad \forall h,~~~ q\leq h<n,
\] 
and $p$ is the number of shortest paths from node $i$ to node $j$, i.e.,
\[
c_{ij}^{(h)}=0,\quad \forall h,~~~ 1< h<q,\quad \mathrm{and}  \quad c_{ij}^{(h)} = 
p,\quad \forall h,~~~ q\leq h<n.
\] 
\end{proposition}

\begin{proposition}\label{fnnz2}
Let $\cal{G}$ be a graph that is not necessarily connected. Fix an integer $K$, with 
$1\leq K<n$. For all $i\ne j$, let $q$, $1\leq q\leq K$, be the smallest positive integer 
such that $a_{ij}^{(q)}=p$. Then there are $p$ shortest paths from node $i$ to node $j$ 
that consist of $q$ edges. Hence, $\ell_{ij}^{(K)}=q$ and $c_{ij}^{(K)} = p$. If there is 
no positive integer $p$ such that $a_{ij}^{(q)}=p$, then there is no path made up of at 
most $K$ edges from node $i$ to node $j$. Then $\ell_{ij}^{(K)} = \infty$ and 
$c_{ij}^{(K)} = 0$.
\end{proposition}

The following result can be shown by using Propositions \ref{fnnz} and \ref{fnnz2}. 

\begin{proposition}\label{fnnz3}
Let $\cal{G}$ be a graph that is not necessarily connected. Then, for all $i\ne j$, the
entries of the $(h-1)$-path length matrix and $(h-1)$-path counter matrix, of the $h$-path 
length matrix and $h$-path counter matrix, and of the path length and path counter 
matrices satisfy
\[
\ell_{ij} \leq \ell_{ij}^{(h)}\leq\ell_{ij}^{(h-1)}\leq \infty  \,\,\,\mathrm{ and } 
\,\,\, 0\leq c_{ij}^{(h-1)} \leq c_{ij}^{(h)}\leq c_{ij},  \quad \forall h,~~~ 1<h<n-1.
\] 
If the graph $\cal{G}$ is connected and has diameter $d_{\mathcal{G}}$, then
\[
\ell_{ij}^{(h)} = \ell_{ij}^{(d_{\mathcal{G}})}=\ell_{ij}<\infty \,\,\,\mathrm{ and } 
\,\,\, c_{ij}^{(h)} = c_{ij}^{(d_{\mathcal{G}})}=c_{ij}>0,  \quad \forall h,~~~ 
d_{\mathcal{G}}\leq h<n.
\] 
\end{proposition}

It is easy to see that the triangle inequality holds for the entries of the path length 
matrix $L$ for a graph $\cal{G}$. The graph is not required to be connected.

\begin{proposition}\label{fnnz5}
Using the notation of Proposition \ref{fnnz3}, we have
$$
\ell_{ij}\leq \ell_{ih}+\ell_{hj},\qquad 1\leq i,j\leq n.
$$
\end{proposition}

Note that the triangle inequality might not hold for the entries of a $K$-path length 
matrix for $1\leq K< n-1$. Specifically, it may happen, for some $1\leq i,j\leq n$, that 
$\ell_{ij}^{(K)}=\infty$, whereas $\ell_{ih}^{(K)}+\ell_{hj}^{(K)}<\infty$.

The MATLAB function \text{ConnectivityMatrices} provided below, with the adjacency matrix 
$A\in\R^{n\times n}$ for a graph $\mathcal{G}$ and the integer $K$, with $1\leq K < n$, as
input, returns the $K$-path length matrix $L^{(K)}$ and the $K$-path counter matrix 
$C^{(K)}$ associated with $\mathcal{G}$. If $K = n-1$,  then the function 
\text{ConnectivityMatrices} returns the path length matrix $L$ and path counter matrix 
$C$. The operator $==$ in the function stands for logical equal to, whereas the operator 
$*$ in line 6 denotes matrix multiplication. The function call ${\sf eye}(n)$ returns the 
identity matrix $I\in\R^{n\times n}$.  For a vector $v\in\R^{n}$, the function 
${\sf diag}(v)$ determines a diagonal matrix whose  $(j,j)$th entry is the $j$th component
of $v$, whereas the function call ${\sf any}(v)$ returns $1$ if the vector $v$ contains a 
non-zero entry and $0$ otherwise. For a matrix $M\in\R^{n\times n}$, the function call 
${\sf size}(M,1)$ returns $n$, the function call ${\sf diag}(M)$ computes a vector whose 
$j$th component is the $(j,j)$th entry of $M$, and the function call ${\sf any}(M)$ gives a 
row vector whose $j$th component is $1$ if the $j$th column of $M$ contains a non-zero 
entry and $0$ otherwise. The command $M(M==0)=\inf$ assigns the value $\infty$ to the zero
entries of $M$.

\begin{algorithm}[ht]
\begin{algorithmic}[1]
\Function{\text{ConnectivityMatrices}}{$A,K$}
\State $C = A$;
\State $n=\text{size}(A,1)$; $h = 2$; 
\State $L=A+\text{eye}(n)$; 
\While {\text{any}(\text{any} ($L==0$)) and $h<=K$} 
 \State $B = A*C$;
 \State $B=B-\text{diag}(\text{diag}($B$))$;
 \For {$i = 1:n$}
     \For {$j = 1:n$}
        \If{$C(i,j) == 0$ and $B(i,j) >0$}
           \State $C(i,j)=B(i,j)$;
           \State $L(i,j)=h$;
        \EndIf
     \EndFor
  \EndFor
  \State $h = h+1$;
\EndWhile
\State $C=C+\text{eye}(n)$;
\State $L(L==0)=\text{inf}$;  
\State $L=L-\text{eye}(n)$; \\
\Return $C$, $L$; 
\EndFunction
\end{algorithmic}
\end{algorithm}

\begin{remark}\label{rmk_1}
The graphs considered in this paper are unweighted, i.e., the edges represent connections
from one node to another and have no associated weights. This condition is necessary for 
Propositions \ref{fnnz}-\ref{fnnz3} to hold. Indeed, in the case of a weighted graph, the 
length of the shortest paths from node $i$ to node $j$ composed of $K+1$ edges, i.e., the 
sum of all weights of the edges of a path, could be smaller than the length of the 
shortest paths from node $i$ to node $j$ made up of at most $K$ edges; see \cite{NR_na} 
for further discussion.
\end{remark}

\section{Global $K$-efficiency, global $K$-connectivity, and related measures}\label{s3}
The \emph{efficiency} of a path between any two nodes of a graph $\mathcal{G}$ is defined 
as the inverse of the length of the path. The sum of the efficiencies of all shortest 
paths starting from node $i$,
\[
h_i=\sum_{j\ne i}\frac{1}{\ell_{ij}},
\]
is referred to as the \emph{harmonic centrality} of node $i$; see, e.g., \cite{BBV}. This
measure gives a large centrality to nodes that have small shortest path lengths to 
the other nodes of the graph. Taking into account only shortest paths made up of at most 
$K$ edges, with $1\leq K<n$, the notions of \emph{harmonic $K_{\rm out}$-centrality} 
$h_i^{K_{\rm out}}$, which considers such shortest paths that start at node $i$, and 
\emph{harmonic $K_{\rm in}$-centrality} $h_i^{K_{\rm in}}$, which considers such shortest 
paths that end at node $i$, were introduced in \cite{NR_na}. The global $K$-efficiency 
\eqref{ekg} can be expressed as 
$$
h_i^{K_{\rm out}}=
\sum_{i\ne j}\frac{1}{\ell^{(K)}_{ij}}, \,\;\; h_i^{K_{\rm in}}=
\sum_{i\ne j}\frac{1}{\ell^{(K)}_{ji}}\,, \;\;\; \\
e^{(K)}_{\mathcal{G}}=\frac{1}{n(n-1)}\sum_{i\in\mathcal{V}} h_i^{K_{\rm out}}=
\frac{1}{n(n-1)}\sum_{i\in\mathcal{V}} h_i^{K_{\rm in}}.
$$
The following result shows how we can estimate the global efficiency; cf. equation 
\eqref{e_G}.

\begin{proposition}\label{ineq_e}
The global $K$-efficiency \eqref{ekg} of the graph 
$\mathcal{G}=\langle\mathcal{V},\mathcal{E}\rangle$ with $|\mathcal{V}|=n$ satisfies, for
$1\leq K<n$, 
$$
e^{(1)}_{\mathcal{G}}\leq e^{(2)}_{\mathcal{G}}\leq \cdots \leq e^{(n-1)}_{\mathcal{G}}=
e_{\mathcal{G}}.
$$
\end{proposition}

\begin{proof}
The proof follows from equations \eqref{ekg} and \eqref{e_G}. 
\end{proof}

The global $K$-efficiency \eqref{ekg} and global efficiency \eqref{e_G} only depend 
on the $K$-path length matrix and the path length matrix, respectively. The shortest paths
between any two distinct nodes are of particular importance for the communicability of a 
network. With the same number of such paths, the shorter the path length, the larger the 
network efficiency. However, with the same path length, the larger the number of such 
paths, the larger the network connectivity. We therefore introduce novel measures that 
also depend on the $K$-path counter matrix. 

\begin{definition}\label{sh_out}
The \emph{connectivity $K_{\rm out}$-centrality} of a node $i$ is the sum of the number of
the shortest paths made up of at most $K$ edges, with $1\leq K<n$, that start at node $i$ 
and are normalized by their length,
$$
\hat{h}_i^{K_{\rm out}}=
\sum_{i\ne j}\frac{c^{(K)}_{ij}}{\ell^{(K)}_{ij}}, \qquad \forall i\in \mathcal{V}.
$$ 
\end{definition}

\begin{definition}\label{sh_in}
The \emph{connectivity $K_{\rm in}$-centrality} of a node $i$ is the sum of the number of
the shortest paths made up of at most $K$ edges, with $1\leq K<n$, that end at node $i$ 
and are normalized by their length. 
$$
\hat{h}_i^{K_{\rm in}}=
\sum_{i\ne j}\frac{c^{(K)}_{ji}}{\ell^{(K)}_{ji}}, \qquad \forall i\in \mathcal{V}.
$$ 
\end{definition}

\begin{proposition}\label{ineq_hi}
For $1\leq\hat{K}<\tilde{K}<n$, the connectivity $K_{\rm out}$- and $K_{\rm in}$- 
centralities of any node $i$ satisfy the inequalities
\begin{equation} \label{inhhat}
{\rm deg}_i^{\rm out}\leq \hat{h}_i^{\hat{K}_{\rm out}}\leq\hat{h}_i^{\tilde{K}_{\rm out}}
\quad \mathrm{ and } \quad {\rm deg}_i^{\rm in}\leq\hat{h}_i^{\hat{K}_{\rm in}}\leq 
\hat{h}_i^{\tilde{K}_{\rm in}}, \qquad \forall i\in \mathcal{V}.
\end{equation}
\end{proposition}

\begin{proof}
It is straightforward to observe that $\hat{h}_i^{1_{\rm out}} = {\rm deg}_i^{\rm out}$ 
and $ \hat{h}_i^{1_{\rm in}} ={\rm deg}_i^{\rm in}$. The proof now follows from 
Definitions \ref{sh_out} and \ref{sh_in}, and Proposition \ref{fnnz3}.
\end{proof}

Let the \emph{$\hat{h}^{K_{\rm in}}$-center} of a graph $\mathcal{G}$ be the set of all 
nodes with the largest connectivity $K_{\rm in}$-centrality and the 
\emph{$\hat{h}^{K_{\rm out}}$-center} be the set of all nodes with the largest 
connectivity $K_{\rm out}$-centrality. 

\begin{remark}\label{rmk_2}
Generally, any edge from a node in the $\hat{h}^{K_{\rm out}}$-center to a node in the
$\hat{h}^{K_{\rm in}}$-center is easily replaceable by a short geodesic since it is 
located in a part of the graph with many shortest paths.
\end{remark}

Several centrality measures have been proposed for edges of complex networks; see, e.g., 
\cite{DMR2,DMR3,GN,NR}. The above comment suggests the following notion of edge centrality
in geodesic connectivity.

\begin{definition}\label{esh}
The \emph{edge $K$-connectivity centrality} of an edge $e(i\rightarrow j)$ is the product
of the connectivity $K_{\rm out}$-centrality of node $i$ and the connectivity 
$K_{\rm in}$-centrality of node $j$:
\begin{equation*}\label{ecc}
\hat{h}_{i,j}^{(K)}=\hat{h}_i^{K_{\rm out}}\,\hat{h}_j^{K_{\rm in}}, \qquad \forall 
e(i\rightarrow j)\in \mathcal{E}.
\end{equation*}
\end{definition}

We refer to $\hat{h}_{i,j}=\hat{h}_{i,j}^{(n-1)}$ as the 
\emph{edge connectivity centrality} of edge $e(i\rightarrow j)$. The following result 
holds for this centrality measure.

\begin{proposition}\label{ineq_he}
For $1\leq K<n$, the edge $K$-connectivity centrality of an edge 
$e(i\rightarrow j)\in \mathcal{E}$ satisfies the inequalities
$$
{\rm deg}_i^{\rm out}{\rm deg}_j^{\rm in}=\hat{h}_{i,j}^{(1)}\leq \hat{h}_{i,j}^{(2)}\leq 
\cdots \leq \hat{h}_{i,j}^{(n-1)}=\hat{h}_{i,j}.
$$
\end{proposition}

\begin{proof}
The proof follows from Proposition \ref{ineq_hi} and Definition \ref{esh}.
\end{proof}

The following results show how we can estimate global connectivity; cf. equation 
\eqref{c_G}.

\begin{proposition}\label{eq_c}
The global $K$-connectivity, for $1\leq K<n$, satisfies 
$$
c^{(K)}_{\mathcal{G}}=\frac{1}{n(n-1)}\sum_{i\in\mathcal{V}} \hat{h}_i^{K_{\rm out}}=
\frac{1}{n(n-1)}\sum_{i\in\mathcal{V}} \hat{h}_i^{K_{\rm in}}.
$$
\end{proposition}

\begin{proof}
The proof follows from equation \eqref{ckG}, and Definitions \ref{sh_out} and 
\ref{sh_in}.
\end{proof}

\begin{proposition}\label{ineq_c}
The global $K$-connectivity of $\mathcal{G}=\langle\mathcal{V},\mathcal{E}\rangle$, with
$|\mathcal{V}|=n$ and $|\mathcal{E}|=m$, satisfies, for $1\leq K<n$, 
$$
\frac{m}{n(n-1)}=c^{(1)}_{\mathcal{G}}\leq c^{(2)}_{\mathcal{G}}\leq \cdots \leq 
c^{(n-1)}_{\mathcal{G}}=c_{\mathcal{G}}.
$$
\end{proposition}

\begin{proof}
We observe that 
$$
\sum_{i\in\mathcal{V}}\hat{h}_i^{1_{\rm in}}=\sum_{i\in\mathcal{V}}{\rm deg}_i^{\rm in}=
\sum_{i\ne j\in\mathcal{V}} a_{ij}=m \quad \quad \mathrm{ and } \quad \quad
\sum_{i\in\mathcal{V}}\hat{h}_i^{1_{\rm out}}=\sum_{i\in\mathcal{V}} {\rm deg}_i^{\rm out}
=\sum_{i\ne j\in\mathcal{V}} a_{ij}=m.
$$
The result now follows from \eqref{inhhat} and Proposition \ref{eq_c}.
\end{proof}

\subsection{Computational complexity}
Let the graph ${\cal G}$ have $n$ nodes. For each $K=2,3,\ldots,n-1$, the evaluation of 
the matrices $C^{(K)}$ and $L^{(K)}$ requires $O(n^3)$ arithmetic floating point 
operations (flops) due to the required multiplication of $n \times n$ matrices. Computing 
the connectivity $K_{\rm in}$- or $K_{\rm out}$-centralities demands $O(n^2)$ flops and 
the calculation of the global $K$-connectivity costs $O(n)$ flops; see Proposition 
\ref{eq_c}.

If one also would like to compute edge connectivity centralities, then this is fairly
inexpensive. The evaluation of the connectivity $K_{\rm in}$- and  
$K_{\rm out}$-centralities requires $O(n^2)$ flops, and the computation the edge 
connectivity centralities requires $m$ flops.

\section{Edge $K$-betweenness centrality versus edge $K$-connectivity centrality}\label{s4}
The betweenness centrality of any node $i$ is defined as the sum over all ordered pairs of
distinct nodes $(h,k)$ of the fraction of shortest paths from node $h$ to node $k$ that 
pass through node $i$ (where $i$ is not an endpoint); see, e.g.,
\cite{estrada2011structure,Fr,newman2010}.

The edge betweenness of any edge $e(i\rightarrow j)$, denoted as $\hat{b}_{i,j}$, is 
defined as the sum over all ordered pairs of distinct nodes $(h,k)$ of the fraction of 
shortest paths from node $h$ to node $k$ that pass through the edge \cite{GN}. The removal
of an edge with high betweenness centrality may affect the communication between many 
pairs of nodes through the shortest paths between them. To capture the role of an edge in
enabling short-range communication in the network, we introduce the notion of edge 
$K$-betweenness centrality, where $K$ is a positive integer such that $1\leq K<n$.

\begin{definition}\label{kb}
The \emph{edge $K$-betweenness centrality} of an edge $e(i\rightarrow j)$, denoted as 
$\hat{b}_{i,j}^{(K)}$, is defined as the sum over all ordered pairs of distinct nodes 
$(h,k)$ of the fraction of shortest paths from node $h$ to node $k$, whose length is 
less than or equal to $K$, that pass through the edge.
\end{definition}

The following result holds.

\begin{proposition}\label{ineq_he2}
For $1\leq K<n$, the edge $K$-betweenness centrality of any edge 
$e(i\rightarrow j)\in\mathcal{E}$ satisfies 
$$
1=\hat{b}_{i,j}^{(1)}\leq \hat{b}_{i,j}^{(2)}\leq \cdots \leq \hat{b}_{i,j}^{(n-1)}=
\hat{b}_{i,j}.
$$
\end{proposition}

\begin{proof}
The proof follows from Definition \ref{kb} and by observing that, when $K=1$, any edge 
has $K$-betweenness centrality $1$.
\end{proof}

\begin{example}\label{ex1}
Consider the graph $\cal{G}$ with $n=5$ nodes and $m=12$ edges depicted in Figure 
\ref{fig1}. The adjacency matrix $A$ as well as the $K$-path length and $K$-path counter 
matrices for $K=1$ and $K=2$ are given by
\begin{equation*}
A=\begin{bmatrix} 
0 \phantom{0} & 1 \phantom{0}& 1 \phantom{0}& 0 \phantom{0}& 0 \\ 
1 \phantom{0}& 0 \phantom{0}& 0 \phantom{0}& 1 \phantom{0}& 0\\ 
1\phantom{0} & 1\phantom{0} & 0\phantom{0} & 1\phantom{0} & 1\\ 
0\phantom{0} & 1\phantom{0} & 1\phantom{0} & 0\phantom{0} & 1\\ 
 0\phantom{0} & 0\phantom{0} & 1\phantom{0} & 0\phantom{0} & 0
\end{bmatrix},
\quad
L^{(1)}=\begin{bmatrix} 
0 & 1 & 1 & \infty &  \infty \\ 
1 & 0 &  \infty& 1 &  \infty\\ 
1 & 1 & 0 & 1 & 1\\ 
 \infty & 1 & 1 & 0 & 1\\ 
  \infty & \infty & 1 &  \infty & 0
 \end{bmatrix},
\quad
C^{(1)}=\begin{bmatrix} 
1\phantom{0} & 1\phantom{0} & 1\phantom{0} & 0\phantom{0} & 0 \\ 
1 \phantom{0}& 1\phantom{0} & 0\phantom{0} & 1\phantom{0} & 0\\ 
1\phantom{0} & 1\phantom{0} & 1\phantom{0} & 1\phantom{0} & 1\\ 
0 \phantom{0}& 1\phantom{0} & 1\phantom{0} & 1\phantom{0} & 1\\ 
 0\phantom{0} &0\phantom{0} & 1\phantom{0} & 0\phantom{0} & 1
 \end{bmatrix},
\end{equation*}
\begin{equation*}
L=L^{(2)}=\begin{bmatrix} 
0\phantom{0} & 1\phantom{0} & 1\phantom{0} & 2\phantom{0} & 2 \\ 
1\phantom{0} & 0\phantom{0} & 2\phantom{0} & 1\phantom{0} & 2\\ 
1\phantom{0} & 1\phantom{0} & 0\phantom{0} & 1\phantom{0} & 1\\ 
2 \phantom{0}& 1\phantom{0} & 1\phantom{0} & 0\phantom{0} & 1\\ 
 2\phantom{0} & 2\phantom{0} & 1\phantom{0} & 2\phantom{0} & 0
\end{bmatrix},
\quad
C=C^{(2)}=\begin{bmatrix} 
1\phantom{0} & 1\phantom{0} & 1\phantom{0} & 2\phantom{0} & 1 \\ 
1\phantom{0} & 1\phantom{0} & 2\phantom{0} & 1\phantom{0} & 1\\ 
1\phantom{0} & 1\phantom{0} & 1\phantom{0} & 1\phantom{0} & 1\\ 
2\phantom{0} & 1\phantom{0} & 1\phantom{0} & 1\phantom{0} & 1\\ 
 1\phantom{0} & 1\phantom{0} & 1\phantom{0} & 1\phantom{0} & 1
\end{bmatrix}.
\end{equation*}
It is immediate to see that $\cal{G}$ is connected and has diameter $d_{\cal{G}}=2$; see 
\eqref{d_G}. One can easily determine the global efficiency \eqref{e_G} and global 
connectivity \eqref{c_G} of $\cal{G}$. We have $e_{\cal{G}} = 0.800$ and 
$c_{\cal{G}} = 0.875$. A rough lower bound for $c_{\cal{G}}$ can be computed by using 
Proposition \ref{ineq_c}. We obtain $c^{(1)}_{\mathcal{G}}=\frac{m}{n(n-1)}=0.6$. The edge
betweenness centralities are reported in Figure \ref{fig1} (left) and the edge 
connectivity centralities in Figure \ref{fig1} (right). 

Let $\tilde{\cal{G}}={\cal{G}}_{ij}$ denote the graph obtained after removing the edge 
$e(i\rightarrow j)$ from ${\cal{G}}$. Table \ref{tab1} summarizes the connectivity losses
incurred by removing selected edges. We observe that the effect on the global connectivity
of removing edge $e(2\rightarrow 1)$ is underestimated by its edge betweenness centrality. 
Conversely, for all the other edge removals, we note that the larger the betweenness 
centrality of the removed edge, the more the diameter of the graph $\tilde{\cal{G}}$ 
increases and the global connectivity $c_{\tilde{\cal{G}}}$ decreases. Table \ref{tab1} 
shows the decrease of the latter measure to correspond to that of the estimate
$c_{\tilde{\cal{G}}}^{(2)}$, which only takes the shortest paths into account that are 
made up of at most $2$ edges, when the diameter is $3$. The table also shows the decrease 
of $c_{\tilde{\cal{G}}}$ to correspond to the decrease of both estimates 
$c_{\tilde{\cal{G}}}^{(2)}$ and $c_{\tilde{\cal{G}}}^{(3)}$ when the graph is 
disconnected. Moreover, for edges with the same betweenness centrality, the risk of 
disconnecting the graph by removing an edge is approximately inversely proportional to the
connectivity centrality of the edge; see Remark \ref{rmk_2}.

Recall that all edges have the same $1$-betweenness centrality, equal to $1$, while their 
$1$-connectivity centrality is given by the product of the out-degree of the node from 
which they start and of the in-degree of the node at which they end; see Propositions 
\ref{ineq_he}-\ref{ineq_he2}. Therefore, a rough approach to gain insight into which edge 
removals may result in a disconnected graph is to look for the smallest product. For
the present example this approach identifies the edge $e(5\rightarrow 3)$ with 
${\rm deg}_5^{\rm out}= 1$ and ${\rm deg}_3^{\rm in}=3$. This is in agreement with the 
previous analysis; see Table \ref{tab1}.

\begin{figure}
\centering
\begin{tabular}{cc}
\includegraphics[scale = 0.4]{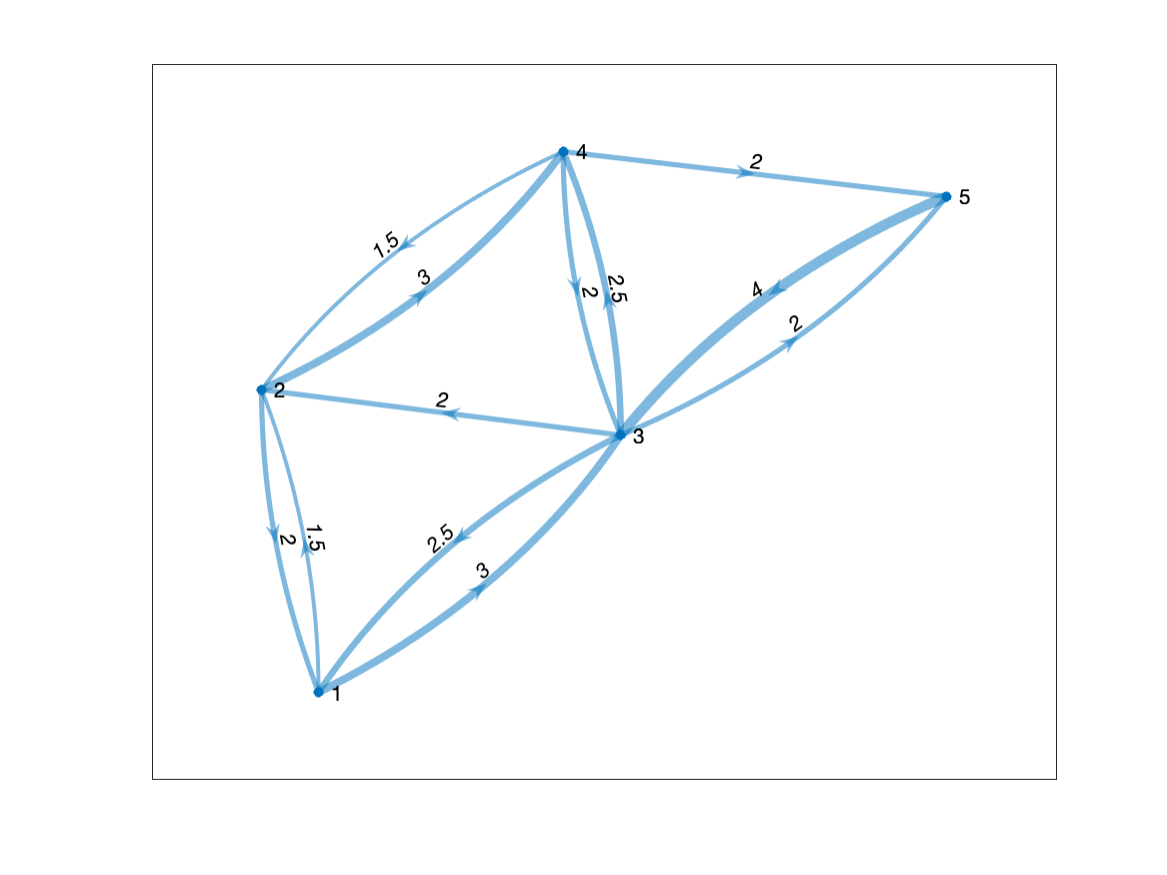} 
\includegraphics[scale = 0.4]{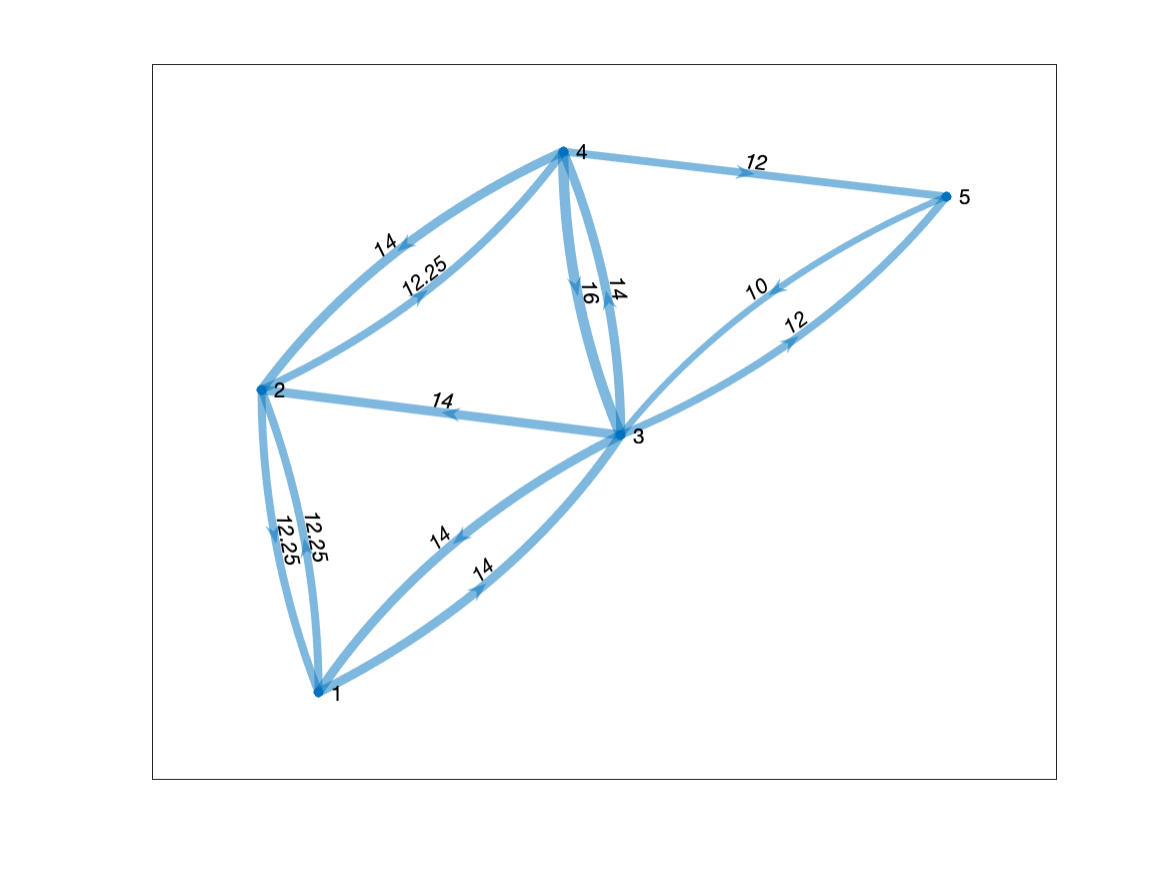} 
\end{tabular}
\caption{Graph considered in Examples \ref{ex1} and \ref{ex2}. Edge betweenness 
centralities (left) and edge connectivity centralities (right) are reported and displayed 
with proportional edge widths.}
\label{fig1}
\end{figure}              

\begin{table}[h!tb]
\centering
\begin{tabular}{c | c c | l | l l l | l l l}
$\tilde{\cal{G}}$&$\hat{b}_{ij}$ & $\hat{h}_{ij}$ & $d_{\tilde{\cal{G}}}$ &$c_{\tilde{\cal{G}}}^{(2)}$  & $c_{\tilde{\cal{G}}}^{(3)}$ &$c_{\tilde{\cal{G}}}$&$e_{\tilde{\cal{G}}}^{(2)}$  & $e_{\tilde{\cal{G}}}^{(3)}$ &$e_{\tilde{\cal{G}}}$\\
\hline
${\cal{G}}_{53}$  & 4.0 & $10.00$ & $\infty$&   $0.7500 $  &   $0.7500 $ & $0.7500$ &   $0.6750 $  &   $0.6750$ & $0.6750$ \ \\
\hline
${\cal{G}}_{13}$  & 3.0 & $14.00$ & $3$&   $0.7500$  & ${0.7833}$ & ${0.7833}$ &   $0.7250 $  &   $0.7583 $ & $0.7583$\ \\
${\cal{G}}_{24}$  & 3.0 & $12.25$ & $3$&   $0.7500$  & ${0.7833}$ & ${0.7833}$ &   $0.7250$  &   $0.7583 $ & $0.7583$ \\
${\cal{G}}_{34}$  & 2.5 & $14.00$ & $3$&   $0.8000$  & $0.8167$ & $0.8167$ &   $0.7500 $  &   $0.7667 $ & $0.7667$\ \\
${\cal{G}}_{31}$  & 2.5 & $14.00$ & $3$&   $0.8000$  & $0.8167$ & $0.8167$&   $0.7500 $  &   $0.7667$ & $0.7667$\ \\
${\cal{G}}_{35}$  & 2.0 & $12.00$ & $3$&   $0.8250$  & $0.8583$ & $0.8583$&   $0.7500 $  &   $0.7667 $ & $0.7667$\ \\
${\cal{G}}_{45}$  & 2.0 & $12.00$ & $3$&   $0.8250 $  & $0.8583$ & $0.8583$&   $0.7500 $  &   $0.7667 $ & $ 0.7667$\ \\
${\cal{G}}_{21}$  & 2.0 & $12.25$ & $3$&   $0.7750$  & $ 0.7917$ & $ 0.7917$&   $0.7500 $  &   $0.7667 $ & $0.7667$\ \\
${\cal{G}}_{32}$  & 2.0 & $14.00$ & $3$&   $0.8500$  & $0.8833$ & $0.8833$&   $0.7500 $  &   $0.7667 $ & $0.7667$ \ \\
\hline
${\cal{G}}_{43}$  & 2.0 & $16.00$ & $2$& $0.8000$  & $0.8000$ & $0.8000$&   $0.7750 $  &   $0.7750 $ & $0.7750$\ \\
${\cal{G}}_{12}$  & 1.5 & $12.25$ & $2$&   $0.8250$  & $0.8250$ & $0.8250$&   $0.7750 $  &   $0.7750$ & $0.7750$\ \\
${\cal{G}}_{42}$  & 1.5 & $14.00$ & $2$&   $0.8250$  & $0.8250$ & $0.8250$&   $0.7750$  &   $0.7750 $ & $0.7750$\ \\
\end{tabular}
\caption{Examples \ref{ex1} and \ref{ex2}. Betweenness and connectivity centralities of each 
removed edge and relevant measures for the resulting graphs, i.e., the diameter, global 
$K$-connectivity, and global $K$-efficiency for $2\leq K\leq 4$.} 
\label{tab1}
\end{table}
\end{example}

\subsection{Global network connectivity after edge-removals}\label{4.1}
Solely focusing on the amount of decrease of the global connectivity caused by the removal
of an edge may be misleading. In fact, this metric might not be consistent with 
expectations, because in some complex networks it may happen that by removing an edge 
$e(i\rightarrow j)$, the global connectivity of the graph 
$\tilde{\cal{G}}={\cal{G}} \backslash\{e(i\rightarrow j)\}$ so obtained is larger than the
global connectivity of the original graph ${\cal{G}}$. {  This phenomenon is known as
the Braess paradox \cite{Br} and has recently been discussed by Altafini et al. 
\cite{ABCMP}.} For example, the following scenario may occur:
\begin{itemize}
\item 
The graph $\tilde{\cal{G}}$ has several shortest paths from node $i$ to node $j$, and 
their number is larger than their length, i.e., $\tilde{c}_{ij} >\tilde{\ell}_{ij}>1$, 
while for ${\cal{G}}$ one has $c_{ij}=\ell_{ij}=1$. Then
\[
\frac{\tilde{c}_{ij}}{\tilde{\ell}_{ij}}>1= \frac{c_{ij}}{\ell_{ij}}.
\]
\item 
All $c_{hk}$ shortest paths in ${\cal{G}}$ from node $h$ to node $k$, for some $(h,k)$, 
include the edge $e(i\rightarrow j)$ and do not contain any of the edges that form the 
$\tilde{c}_{ij}$ shortest paths above. Then the shortest paths from node $h$ to node $k$ 
in $\tilde{\cal{G}}$ replace $e(i\rightarrow j)$ with the $\tilde{c}_{ij}$ shortest paths 
above. This yields $\tilde{c}_{hk}=c_{hk}\cdot \tilde{c}_{ij}$ shortest paths of length 
$\tilde{\ell}_{hk}=\ell_{hk}-1+\tilde{\ell}_{ij}$. Therefore, we have 
\[
\frac{\tilde{c}_{hk}}{\tilde{\ell}_{hk}}>\frac{c_{hk}}{\ell_{hk}},
\]
since
\[
\frac{\tilde{c}_{ij}}{\ell_{hk}-1+\tilde{\ell}_{ij}}>\frac{1}{\ell_{hk}},
\]
and $(\tilde{c}_{ij}-1) \ell_{hk}>\tilde{\ell}_{ij}-1$, because 
$\tilde{c}_{ij} >\tilde{\ell}_{ij}>1$ and $\ell_{hk}>0$. 
\end{itemize}

The following result shows how we can check whether the component of global connectivity 
exclusively due to shortest paths of length $2$ is larger for the graph $\tilde{\cal{G}}$ 
than for the graph ${\cal{G}}$.

\begin{proposition}
The global connectivity component due to the shortest paths of length $2$ for the graph 
$\tilde{\cal{G}}$, obtained from the graph ${\cal{G}}$ by removing an edge, is larger than
the same component for $c_{\cal{G}}$ if 
\begin{equation} \label{remij}
c^{(2)}_{\mathcal{G}}-c^{(2)}_{\tilde{\cal{G}}}<\frac{1}{n(n-1)}.
\end{equation}
\end{proposition}

\begin{proof}
We observe that 
$$
c^{(2)}_{\mathcal{G}}-c^{(1)}_{\mathcal{G}}= c^{(2)}_{\mathcal{G}}-\frac{m}{n(n-1)} \quad 
\mathrm{ and }  \quad c^{(2)}_{\tilde{\cal{G}}}-c^{(1)}_{\tilde{\cal{G}}}= 
c^{(2)}_{\tilde{\cal{G}}}-\frac{m-1}{n(n-1)};
$$
see Proposition \ref{ineq_c}. Hence,
$c^{(2)}_{\mathcal{G}}-c^{(1)}_{\mathcal{G}}> 
c^{(2)}_{\tilde{\cal{G}}}-c^{(1)}_{\tilde{\cal{G}}}$
if \eqref{remij} holds. This concludes the proof.
\end{proof}

\begin{remark} \label{rmk_3}
Since checking the global connectivity $c_{\tilde{\cal{G}}}$ of 
$\tilde{\cal{G}}={\cal{G}} \backslash e(i\rightarrow j)$ may not tell us  {the effective loss
of connectivity of the graph,} especially if the component due to shortest paths of length $2$ is 
increased compared to that of $c_{{\cal{G}}}$, one may alternatively compare the global 
efficiency before and after edge removal. Indeed, $e_{\tilde{\cal{G}}}$ cannot be larger 
than $e_{{\cal{G}}}$, because $\tilde{\ell}_{ij}\geq {\ell}_{ij}$; see \eqref{e_G}.
\end{remark}

\begin{example}\label{ex2}
Consider the graph $\cal{G}$ depicted in Figure \ref{fig1}. Table \ref{tab1} shows the 
global connectivity component for graphs $\tilde{\cal{G}}={\cal{G}}_{ij}$ due to the 
shortest paths of length $2$, i.e., $c^{(2)}_{\tilde{\cal{G}}}-c^{(1)}_{\tilde{\cal{G}}}$,
to be smaller than or equal to the corresponding global connectivity component 
$c^{(2)}_{\mathcal{G}}-c^{(1)}_{\mathcal{G}}$ for $\cal{G}$ with one exception. In fact, 
since $c_{\mathcal{G}}=c^{(2)}_{\mathcal{G}}=0.8750$ and $1/(n(n-1))=0.0500$, we can 
verify that equation \eqref{remij} holds for $\tilde{\cal{G}}={\cal{G}}_{32}$, with 
$c^{(2)}_{\tilde{\cal{G}}}=0.8500$. However, Table \ref{tab1} shows that the decrease in 
global efficiency of all graphs $\tilde{\cal{G}}$ is consistent with the decrease in their
global connectivity. Also, the effect on the global efficiency caused by removing the edge
$e(2\rightarrow 1)$ is seen to be consistent with its edge betweenness centrality.
\end{example}

\subsection{Girvan-Newman-type algorithms for the removal of edges until the graph is 
disconnected}\label{4.2}
Without taking into account the metric associated with global connectivity, we can make a 
comparison between the loss of connectivity of a graph by sequentially removing one edge
at a time, with recalculation. The selected edge to be removed has the highest betweenness 
centrality, as in the Girvan-Newman algorithm \cite{GN}, or the selection of the edge also
takes the connectivity centrality into account. Given the adjacency matrix $A$ for a 
connected graph $\cal G$, two algorithms for removing edges until the resulting graph is 
disconnected may be summarized as follows. Differently from the Girvin-Newman algorithm,
Algorithm 1 terminates when the graph obtained by successive edge removals is disconnected.

\vspace{0.2 in}

\noindent {\bf Algorithm 1}:

\noindent   
While $\cal G$ is connected, repeat Steps 1-2:
\begin{enumerate}
\item Evaluate the matrix $L$ and compute the betweenness centrality for all existing 
edges.
\item Remove an edge with the largest betweenness centrality.
\end{enumerate}

\vspace{0.2 in}

\noindent {\bf Algorithm 2}:

\noindent
While $\cal G$ is connected, repeat Steps 1-2: 
\begin{enumerate}
\item Evaluate the matrices $L$ and $C$, and compute the betweenness and connectivity 
centralities for all existing edges.
\item Remove the edge with the highest betweenness centrality and, if there is more than 
one such edge, then remove the one with the smallest connectivity centrality.
\end{enumerate}

\vspace{0.2 in}

\begin{remark}\label{rmk_4}
For large graphs, one has to evaluate the matrices $L$ and $C$ to check whether an 
edge removal has made the graph disconnected. It is not sufficient to determine the 
matrices $L^{K}$ and $C^{K}$. However, measures of $K$-connectivity and $K$-betweenness, 
with $K$ equal to the diameter of the initial graph, generally approximate the relevant 
measures of connectivity and betweenness quite well.
\end{remark}

\begin{remark}\label{rmk_5}
When the graph ${\cal G}$ is undirected, the above algorithms have to be modified: the two
edges with the largest betweenness centrality are removed in each iteration (and, in 
Algorithm 2, if there are more than one such edge pair, then the edge pair with the 
smallest connectivity centrality should be removed). This procedure will select edges
pairs $\{e(i\rightarrow j),e(j\rightarrow i)\}$ that should be removed for suitable 
$i\ne j$. Moreover, as for the measures of connectivity, one could choose to work with the
strictly triangular parts of the matrices $L$ and $C$.
\end{remark}

\section{Numerical illustrations} \label{s5}
The numerical examples reported in this paper have been carried out using MATLAB R2024b on
a 3.2 GHz Intel Core i7 6 core iMac, i.e., with about $15$ significant decimal digits.

The simulation of malicious attacks aimed to disconnect a network may be based on an 
edge-betweenness strategy; see, e.g., \cite{GN,LPWS}. The removal of edges with the
largest betweenness centrality, indeed, is likely to disconnect a graph, because these
edges might represent a ``bridge'' between two parts of a graph. 

\begin{example}\label{ex3}
Consider an undirected connected graph ${\mathcal{G}}$ that represents the German highway 
system network {\em Autobahn}. The graph is available at \cite{air500}. Its $1168$ nodes 
are German locations and its $1243$ edges represent highway segments that connect them. 
The diameter $d_{\cal{G}}$ is $62$, the global connectivity 
$c_{\cal{G}}=c_{\cal{G}}^{(62)}$ is $8.7405\cdot 10^{-2}$, and the global 
efficiency $e_{\cal{G}}=e_{\cal{G}}^{(62)}$ is $6.7175\cdot 10^{-2}$. The first edge pair
removal suggested by Algorithm 1 (see Remark \ref{rmk_5}) is
$\{e(565\rightarrow 219),e(219\rightarrow 565)\}$. A total $20$ edge pair removals 
relative to the initial graph (in the first row) and relative to all graphs that have 
undergone {\it all} preceding edge pair removals are shown in Table \ref{tab3} in the 
order suggested by Algorithm 1 until obtaining a disconnected graph. In more detail,
Table \ref{tab3} displays for each graph obtained from ${\cal G}$ by one or several edge 
removals the diameter, the global connectivity, and the global efficiency. Moreover, the 
table reports the global connectivity and the global efficiency considering shortest paths
of  {length not larger than the diameter $d_{\cal{G}}$ of the given graph $\cal{G}$.} In
addition, we report for each graph the edge pairs with the largest betweenness centrality. 
No situation when there are more than one edge pair with the same betweenness centrality 
occured. Therefore, Algorithm 2 does not come into play for the {\em Autobahn} network. 

\begin{table}[h!tb]
\centering
\begin{tabular}{c | c  c  c  c | c c  c }
$d_{\cal{G}}$& $c_{\cal{G}}$&$e_{\cal{G}}$ & $c^{(62)}_{\cal{G}}$&$e^{(62)}_{\cal{G}}$ &$e(i\leftrightarrow j)$ & $\hat{b}_{ij}$ & $\hat{b}^{(62)}_{ij}$\\
\hline
$62$& $8.7405\cdot 10^{-2}$&$6.7175\cdot 10^{-2}$ & $8.7405\cdot 10^{-2}$&$6.7175\cdot 10^{-2}$ &$e(565\leftrightarrow  219)$ & $9.0512\cdot 10^{4}$ & $9.0512\cdot 10^{4}$ \\

$62$& $8.3116\cdot 10^{-2}$&$6.6176\cdot 10^{-2}$ & $8.3116\cdot 10^{-2}$&$6.6176\cdot 10^{-2}$ &$e(450\leftrightarrow  331)$ & $1.0434\cdot 10^{5}$ & $1.0434\cdot 10^{5}$ \\

$62$& $7.9966\cdot 10^{-2}$&$6.5508\cdot 10^{-2}$ & $7.9966\cdot 10^{-2}$&$6.5508\cdot 10^{-2}$ &$e(393\leftrightarrow  373)$ & $8.8435\cdot 10^{4}$ & $8.8435\cdot 10^{4}$ \\

$62$& $7.7388\cdot 10^{-2}$&$6.3784\cdot 10^{-2}$ & $7.7388\cdot 10^{-2}$&$6.3784\cdot 10^{-2}$ &$e(751\leftrightarrow  735)$ & $8.5493\cdot 10^{4}$ & $8.5493\cdot 10^{4}$ \\

$65$& ${\it 7.7465\cdot 10^{-2}}$&$6.2914\cdot 10^{-2}$ & ${\it 7.7464\cdot 10^{-2}}$&$6.2914\cdot 10^{-2}$ &$e(637\leftrightarrow  619)$ & $7.1084\cdot 10^{4}$ & $7.1084\cdot 10^{4}$ \\

$65$& $7.6442\cdot 10^{-2}$&$6.1956\cdot 10^{-2}$ & $7.6440\cdot 10^{-2}$&$6.1956\cdot 10^{-2}$ &$e(347\leftrightarrow  331)$ & $9.6298\cdot 10^{4}$ & $9.6287\cdot 10^{4}$ \\

$66$& $7.3119\cdot 10^{-2}$&$6.1161\cdot 10^{-2}$ & $7.3119\cdot 10^{-2}$&$6.1160\cdot 10^{-2}$ &$e(729\leftrightarrow  698)$ & $6.3999\cdot 10^{4}$ & $6.3999\cdot 10^{4}$ \\

$66$& ${\it 7.3647\cdot 10^{-2}}$&$6.0600\cdot 10^{-2}$ & ${\it 7.3646\cdot 10^{-2}}$&$6.0599\cdot 10^{-2}$ &$e(538\leftrightarrow 534)$ & $8.4349\cdot 10^{4}$ & $8.4348\cdot 10^{4}$ \\

$66$& $7.0991\cdot 10^{-2}$&$5.9749\cdot 10^{-2}$ & $7.0990\cdot 10^{-2}$&$5.9749\cdot 10^{-2}$ &$e(763\leftrightarrow 698)$ & $1.1038\cdot 10^{5}$ & $1.1038\cdot 10^{5}$ \\

$66$& $6.7411\cdot 10^{-2}$&$5.8279\cdot 10^{-2}$ & $6.7411\cdot 10^{-2}$&$5.8279\cdot 10^{-2}$ &$e(758\leftrightarrow 747)$ & $7.8654\cdot 10^{4}$ & $7.8654\cdot 10^{4}$ \\

$66$& $6.5520\cdot 10^{-2}$&$5.7368\cdot 10^{-2}$ & $6.5520\cdot 10^{-2}$&$5.7367\cdot 10^{-2}$ &$e(763\leftrightarrow 719)$ & $8.0332\cdot 10^{4}$ & $8.0332\cdot 10^{4}$ \\

$66$& $6.4440\cdot 10^{-2}$&$5.6355\cdot 10^{-2}$ & $6.4439\cdot 10^{-2}$&$5.6355\cdot 10^{-2}$ &$e(436\leftrightarrow 410)$ & $8.3411\cdot 10^{4}$ & $8.3402\cdot 10^{4}$ \\

$66$& $6.3704\cdot 10^{-2}$&$5.5847\cdot 10^{-2}$ & $6.3703\cdot 10^{-2}$&$5.5846\cdot 10^{-2}$ &$e(412\leftrightarrow 347)$ & $6.2831\cdot 10^{4}$ & $6.2818\cdot 10^{4}$ \\

$69$& $6.2767\cdot 10^{-2}$&$5.4966\cdot 10^{-2}$ & $6.2765\cdot 10^{-2}$&$5.4965\cdot 10^{-2}$ &$e(892\leftrightarrow 886)$ & $7.1236\cdot 10^{4}$ & $7.1235\cdot 10^{4}$ \\

$74$& $6.1981\cdot 10^{-2}$&$5.3763\cdot 10^{-2}$ & $6.1974\cdot 10^{-2}$&$5.3759\cdot 10^{-2}$ &$e(549\leftrightarrow 543)$ & $7.1317\cdot 10^{4}$ & $7.1317\cdot 10^{4}$ \\

$74$& $6.1788\cdot 10^{-2}$&$5.3207\cdot 10^{-2}$ & $6.1781\cdot 10^{-2}$&$5.3203\cdot 10^{-2}$ &$e(564\leftrightarrow 556)$ & $9.0756\cdot 10^{4}$ & $9.0678\cdot 10^{4}$ \\

$74$& $6.0755\cdot 10^{-2}$&$5.2516\cdot 10^{-2}$ & $6.0750\cdot 10^{-2}$&$5.2511\cdot 10^{-2}$ &$e(255\leftrightarrow 252)$ & $7.2351\cdot 10^{4}$ & $7.2193\cdot 10^{4}$ \\

$78$& $5.8708\cdot 10^{-2}$&$5.1530\cdot 10^{-2}$ & $5.8681\cdot 10^{-2}$&$5.1504\cdot 10^{-2}$ &$e(599\leftrightarrow 588)$ & $8.5994\cdot 10^{4}$ & $8.5812\cdot 10^{4}$ \\

$78$& $5.7390\cdot 10^{-2}$&$5.0591\cdot 10^{-2}$ & $5.7328\cdot 10^{-2}$&$5.0534\cdot 10^{-2}$ &$e(452\leftrightarrow 431)$ & $1.1848\cdot 10^{5}$ & $1.1591\cdot 10^{5}$ \\

$86$& $5.6558\cdot 10^{-2}$&$4.9681 \cdot 10^{-2}$ & $5.6289\cdot 10^{-2}$&$4.9423\cdot 10^{-2}$ &$e(902\leftrightarrow 882)$ & $1.3856\cdot 10^{5}$ & $1.2684\cdot 10^{5}$ \\

$\infty$&  {$5.1455\cdot 10^{-2}$}& {$4.4967\cdot 10^{-2}$} &  {$5.1454\cdot 10^{-2}$}& {$4.4966\cdot 10^{-2}$} & & &  \\
\end{tabular}
\caption{Example \ref{ex3}. Connectivity measures for all graphs determined by Algorithm 
1. The notation $e(i\leftrightarrow j)$ stands for both edges $e(i\rightarrow j)$ and 
$e(j\rightarrow i)$.} 
\label{tab3}
\end{table}

In Table \ref{tab3} we can observe that for each consecutive graph, the global efficiency 
decreases and, apart from two exceptions marked in italics, the same happens for the global 
connectivity. The fact that for each graph, the values of the measures $c_{\cal{G}}$ and 
$e_{\cal{G}}$ are very close suggests that there are not many shortest paths 
joining the pairs of nodes that define the highway segment to be removed. Furthermore, for
each graph we can observe that, for $K=62$, the values of both the global connectivity and 
the global efficiency are extremely close to the value of the global $K$-connectivity 
\eqref{ckG} and to the global $K$-efficiency \eqref{ekg}, respectively; see Remark 
\ref{rmk_4}. Finally, for each graph considered, the edges with the highest betweenness 
centrality connect the same nodes in both directions. These edges also have the highest 
$K$-betweenness centrality, with $K=62$, as is reported in Table \ref{tab3}; see Remark 
\ref{rmk_4}. Figure \ref{fig2} shows the original {\em Autobahn} network and the 
disconnected network obtained  {by removing edges with Algorithm 1 from the original
network until a disconnected network is obtained}.

\begin{figure}
\centering
\begin{tabular}{cc}
\includegraphics[scale = 0.4]{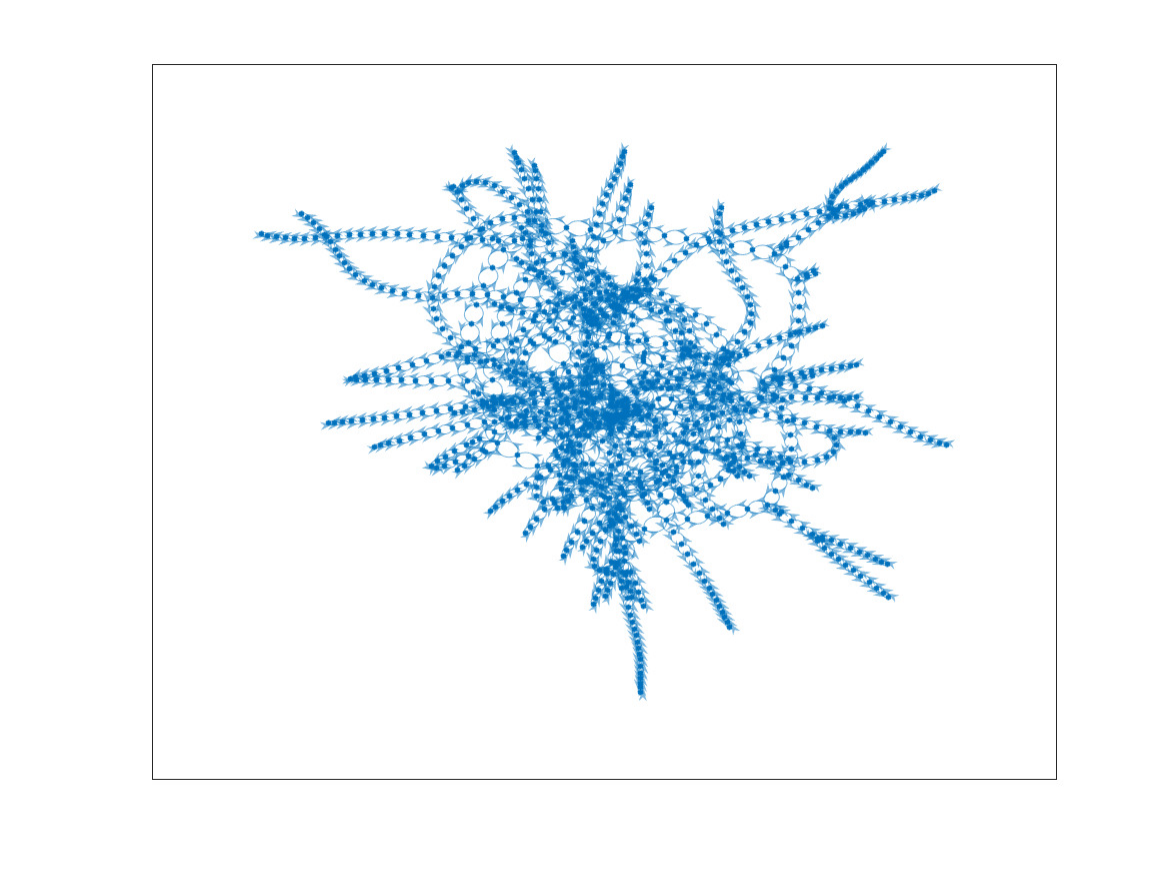} 
\includegraphics[scale = 0.4]{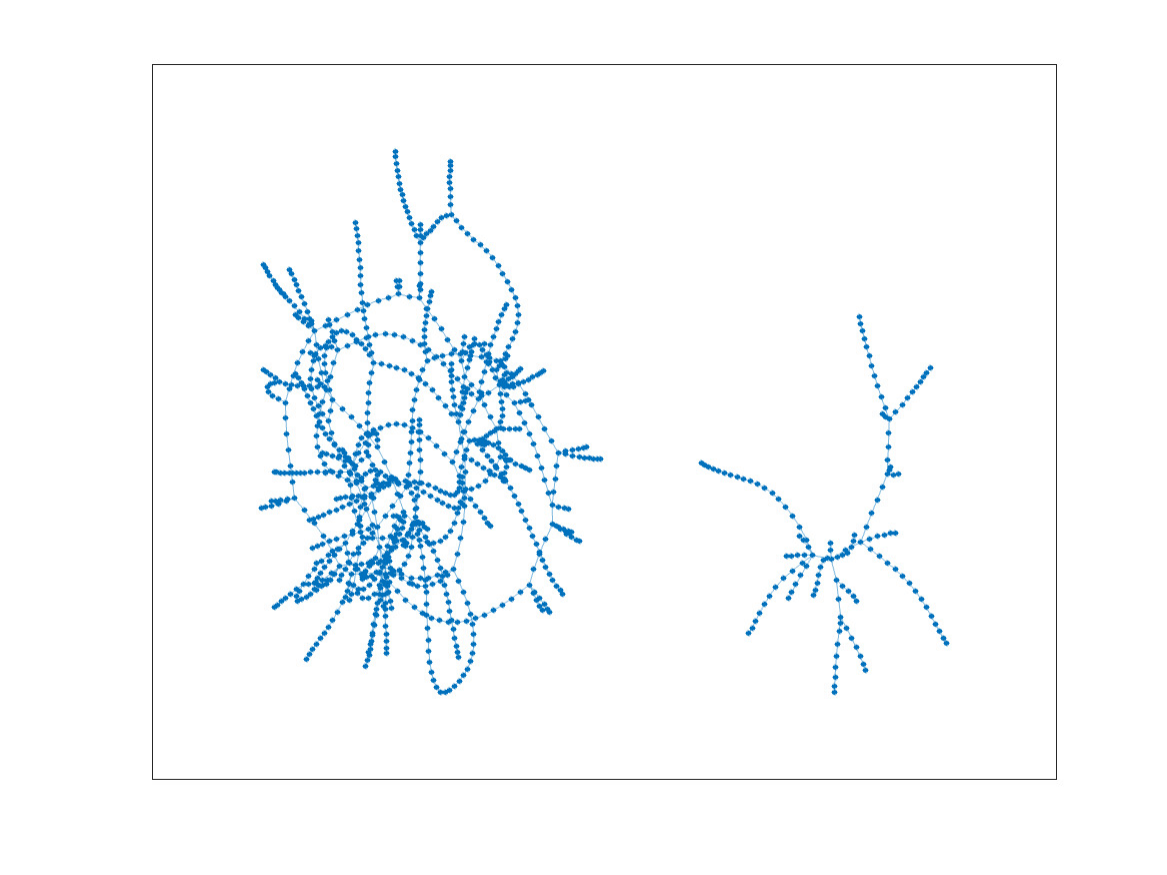} 
\end{tabular}
\caption{Graphs for Example \ref{ex3}: Original graph (left) and disconnected graph  determined
by edge removal with Algorithm 1 (right).}
\label{fig2}
\end{figure}    
\end{example}

\begin{remark}\label{rmk_6}
In the {\em Autobahn} network of Example \ref{ex3}, there is only one shortest path of 
length $62$. It connects node $116$ (Bissendorf) and node $1154$ (W\"ustenbrand). The graph 
can be disconnected by removal of just one undirected edge between node $116$ and node 
$120$ (B\"unde), or by removing one undirected edge between node $1154$ and node $1145$ 
(Wommen). This would decrease the global connectivity from 
$8.7405\cdot 10^{-2}$ to $8.7342\cdot 10^{-2}$ in the first case and to 
$8.7336\cdot 10^{-2}$ in the second case, and it would reduce the global efficiency from 
$6.7175\cdot 10^{-2}$ to $6.7118\cdot 10^{-2}$ in the first case and to 
$6.7126\cdot 10^{-2}$ in the second case. 
However, none of the above removals is suggested by the Girvan-Newman algorithm \cite{GN}, 
which tries to remove the edges that are, in some sense, the most important ones in the 
network. We are not interested in disconnecting the graph after as few edge removals
as possible, but rather in evaluating whether the impact of the removals suggested by the 
Girvan-Newman algorithm is consistent with the progressive decrease of the graph connectivity 
measures.
\end{remark}
 
\begin{example}\label{ex4}
Consider the network \emph{Air500} in \cite{air500}. This dataset describes flight
connections for the top 500 airports worldwide based on total passenger volume. The 
network is based on flight connections between airports within one year from July 1, 2007, 
to June 30, 2008. The network is represented by a connected graph $\cal{G}$ with $500$ 
nodes and $24009$ edges. The nodes represent airports and the edges represent direct 
flight routes between two airports. The graph is directed because some flight routes are
available only in one direction.

The diameter of the graph $\cal{G}$ is $5$, its global connectivity 
$c_{\cal{G}}=c_{\cal{G}}^{(5)}$ is $1.0330\cdot 10^{1}$, and its global efficiency 
$e_{\cal{G}}=e_{\cal{G}}^{(5)}$ is $4.8392\cdot 10^{-1}$. The first two edges to be 
removed, as suggested by Algorithm 1, are, in order, 
\begin{enumerate}
\item remove edge $e(80\rightarrow 177)$, with $\hat{b}_{80,177}=5.2048\cdot 10^{2}$;
\item remove edge $e(177\rightarrow 80)$, with $\hat{b}_{177,80}=5.0446\cdot 10^{2}$.
\end{enumerate}
The first edge-removal decreases the global efficiency to $4.8387\cdot 10^{-1}$ and 
increases the diameter to $6$ and the global connectivity to $1.0341\cdot 10^{1}$. The 
second edge-removal decreases the global efficiency to $4.8383\cdot 10^{-1}$, leaves the 
diameter at $6$, and increases the global connectivity to $1.0347\cdot 10^{1}$. The fact 
that the edge removal increases the global connectivity suggests that for this example 
it makes no sense to consider this indicator; see Subsection \ref{4.1}. We therefore only
will monitor how the global efficiency of the network decreases with the removal of edges.
 {The fact that the connectivity increases when removing an edge is an illustration of 
the Braess paradox; see \cite{ABCMP,Br}.} 

When determining the third edge to be removed, we have to use Algorithm 2 because there 
are four choices suggested by Algorithm 1 with the same betweenness centrality. We order 
these choices in ascending order of connectivity centrality:
\begin{enumerate}
\item[$3_a.$] remove edge $e(10\rightarrow 313)$, with  $\hat{h}_{10,313}=7.5755\cdot 10^{6}$;
\item[$3_b.$] remove edge $e(313 \rightarrow 10)$, with $\hat{h}_{313,10}= 7.6187\cdot 10^{6}$;
\item[$3_c.$] remove edge $e(80 \rightarrow 399)$, with $\hat{h}_{80,399}= 1.5454\cdot 10^{7}$; 
\item[$3_d.$] remove edge $e(399 \rightarrow 80)$, with $\hat{h}_{399.80}= 1.5519\cdot 10^{7}$.
\end{enumerate}

Any of these edge removals disconnect the graph that already has undergone the removal of 
two edges. Table \ref{tab4} reports centrality measures for each one of the four edges and
the effect of their removal on the global efficiency.

\begin{table}[h!tb]
\centering
\begin{tabular}{c | c  c  c  c | c  c c  c } 
$e(i\rightarrow j)$& $\hat{b}_{ij}$ & $\hat{b}^{(5)}_{ij}$& $\hat{h}_{ij}$ & $\hat{h}^{(5)}_{ij}$ &$e_{\cal{G}}$& $e^{(5)}_{\cal{G}}$ \\
\hline
$e(10\rightarrow 313)$ & $4.9900\cdot 10^{2}$ & $4.9900\cdot 10^{2}$& $7.5755\cdot 10^{6}$ & $7.5755\cdot 10^{6}$ & $4.8321\cdot 10^{-1}$&$4.8321\cdot 10^{-1}$ \\
$e(313 \rightarrow 10)$& $4.9900\cdot 10^{2}$& $4.9900\cdot 10^{2}$&$7.6187\cdot 10^{6}$ &  $7.6187\cdot 10^{6}$ & $4.8322\cdot 10^{-1}$&$4.8321\cdot 10^{-1}$ \\
$e(80 \rightarrow 399)$& $4.9900\cdot 10^{2}$& $4.9600\cdot 10^{2}$&$1.5454\cdot 10^{7}$ &  $1.4772\cdot 10^{7}$ & $4.8332\cdot 10^{-1}$&$4.8331\cdot 10^{-1}$ \\
$e(399 \rightarrow 80)$& $4.9900\cdot 10^{2}$& $4.9500\cdot 10^{2}$&$1.5519\cdot 10^{7}$ &  $1.4593\cdot 10^{7}$ & $4.8332\cdot 10^{-1}$&$4.8332\cdot 10^{-1}$ \\
\end{tabular}
\caption{Example \ref{ex3}. Edge-centrality measures and global efficiency for the graph 
after edge removal.} 
\label{tab4}
\end{table}

Table \ref{tab4} shows that the choice $3_a$ results in the maximal loss of global 
efficiency. Note that all the values of global $K$-efficiency, with $K=d_{\cal{G}}=5$, are
extremely close to the corresponding values of global efficiency; see Remark \ref{rmk_4}. 
Furthermore, it can be observed that if one would use the $K$-betweenness values to decide 
on the third edge-removal, the choice would be $e(10\rightarrow 313)$, with the observable 
effect of the largest loss of global $K$-efficiency (equal to that due to the removal of 
$e(313 \rightarrow 10)$).

Finally, we note that starting with the original graph $\cal{G}$, the removal of any one 
of the edges identified in $3_a$, $3_b$, $3_c$, or $3_d$, would have immediately 
disconnected the graph. In detail, the choice $3_a$ would have reduced the global 
efficiency $e_{\cal{G}}$ from $4.8383\cdot 10^{-1}$ to $4.8331\cdot 10^{-1}$, the choice 
$3_b$ to $4.8332\cdot 10^{-1}$, and the choices $3_c$ and $3_d$ to $4.8340\cdot 10^{-1}$. 
In fact, in the Air500 network, the Buenos Aires AEP airport (node $10$) only handles 
routes to/from the Montevideo MVD airport (node $313$) and the Rio de Janeiro SDU airport 
(node $399$) only handles routes to/from the Sao Paulo CGH airport (node $177$). Thus,
if we remove edge $e(10\rightarrow 313)$, then no destination can reached from the AEP 
airport, while if we remove edge $e(313\rightarrow 10)$, then it is impossible to reach 
the AEP airport. Similarly, if we remove edge $e(80\rightarrow 399)$, then it is not 
possible to reach the SDU airport, while if we remove edge $e(399\rightarrow 80)$, then no
destination can be reached from the SDU airport. However, as observed in Remark 
\ref{rmk_6}, our main interest is in the impact of edge removals with the largest 
betweenness centrality.
\end{example}

\begin{example} \label{ex5}
 {Consider the Zachary Karate Club social network, which models friendships among 
members of a US university karate club in the 1970s. In Zachary's original paper 
\cite{Za}, the  network is undirected, connected, and unweighted  with $n=34$ nodes, which
correspond to club members, and $m=78$ edges, each of which corresponds to interactions
of two club members outside the club.  The dataset Zachary78 is available at \cite{SSMC}. 
However, the interaction between nodes $34$ and $23$ is reported ambiguously in \cite{Za}
and the $77$-edge dataset Zachary77 does not contain the edge $e(34 \leftrightarrow 23)$; 
see \cite{ICON}. The diameter of both graphs is $5$.}

 {We apply Algorithm 2 to both graphs Zachary78 and Zachary77. When using
this algorithm to remove $11$ edge pairs from each dataset, the same edge pairs are
eliminated in the same order in both graphs. This yields disconnected graphs. The
graphs determined by $10$ edge pair removals in each dataset are connected. In the only 
step (step 10) when Algorithm 2 finds two edges ($e(14 \leftrightarrow 3)$ and 
$e(8 \leftrightarrow 3)$) with the same highest betweenness centrality, 
these edges also have the same connectivity centrality. In fact, removing either one
of the edges $e(14 \leftrightarrow 3)$ or $e(8 \leftrightarrow 3)$ gives graphs with the 
same global efficiency ($0.4085$ for Zachary78 and $0.4064$ for Zachary77). The global
connectivity for the graphs obtained is $0.8107$ for Zachary78 and $0.7705$ for Zachary77. 
Furthermore, in both cases the next step removes the other one of the these two edges and
gives a disconnected graph. The disconnected components of the graph are made up of 
club members that are friendly to each other.}

 {Interestingly, Algorithm 2 removes the same first $11$ edge pairs if the removal is 
based on the $K$-betweenness centralities and $K$-connectivity centralities with $K=5$; see
Remark \ref{rmk_4}. The use of the $K$-betweenness centralities for a fairly small value
of $K$ reduces the computational burden.}

 {Figure \ref{fig3} shows the original Zachary78 and Zachary77 networks as well as the 
disconnected networks obtained by Algorithm 2 when applied to the original networks until 
they are disconnected. Figure \ref{fig4} displays the global efficiency and global 
connectivity of the two karate club networks as a function of the progressive steps of 
Algorithm 2. The horizontal axis ranges from 0 (original graph) to 11 (after eleven 
steps). The global connectivity is seen not to decrease monotonically.}

\begin{figure}
\centering
\begin{tabular}{cc}
\includegraphics[scale = 0.4]{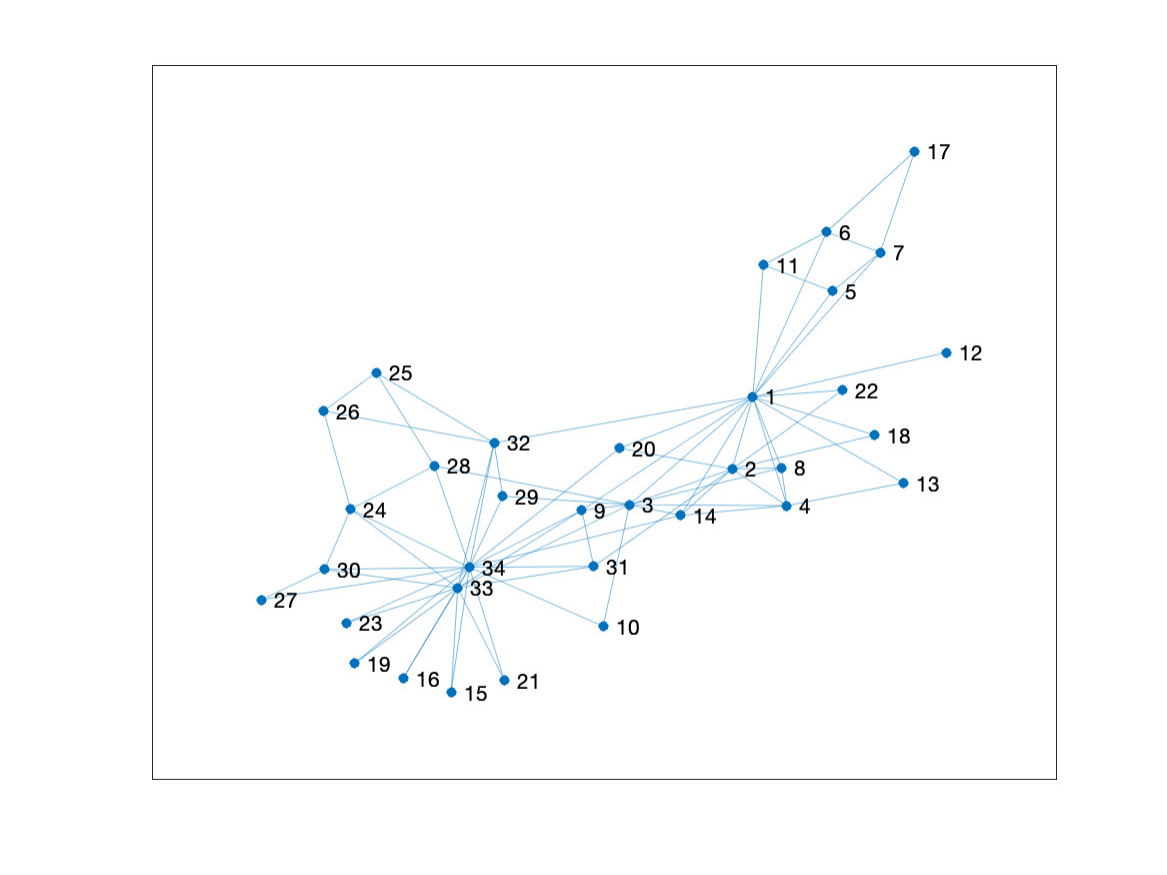} 
\includegraphics[scale = 0.4]{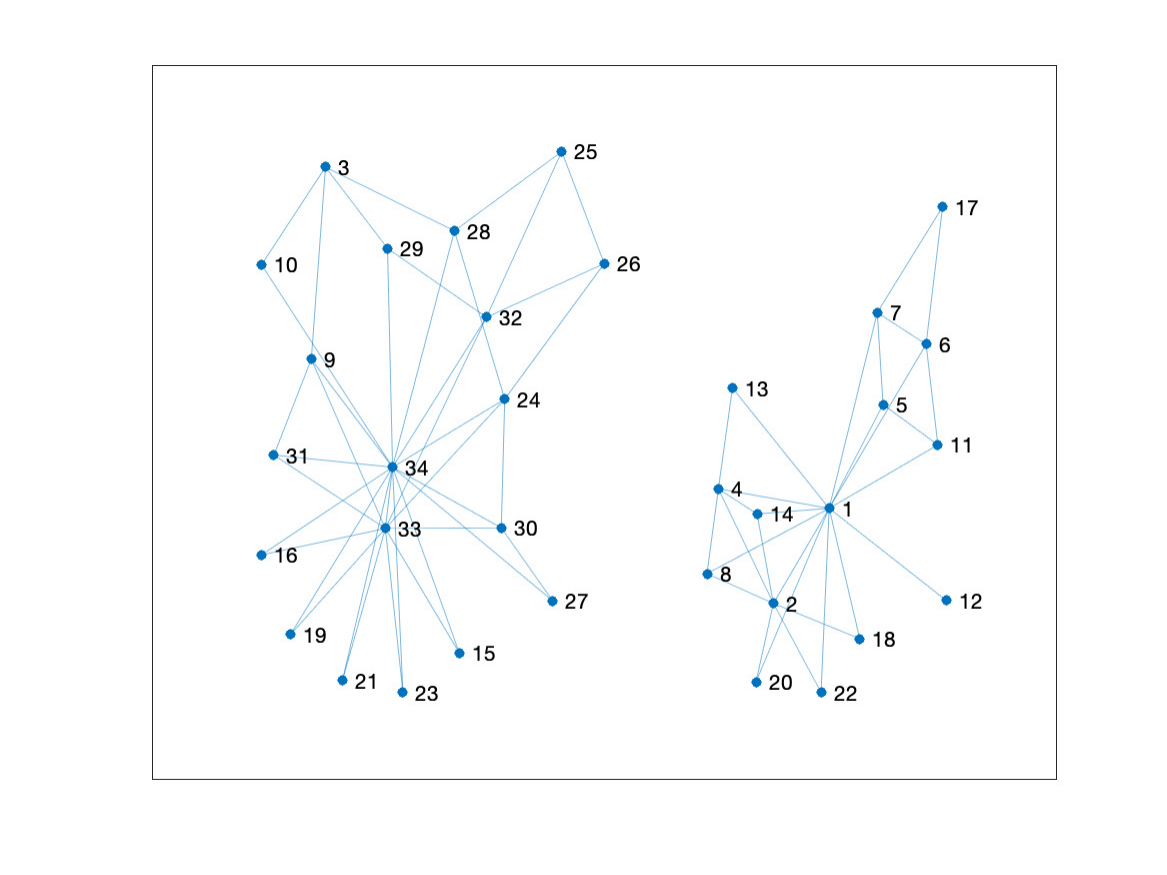} \\
\includegraphics[scale = 0.4]{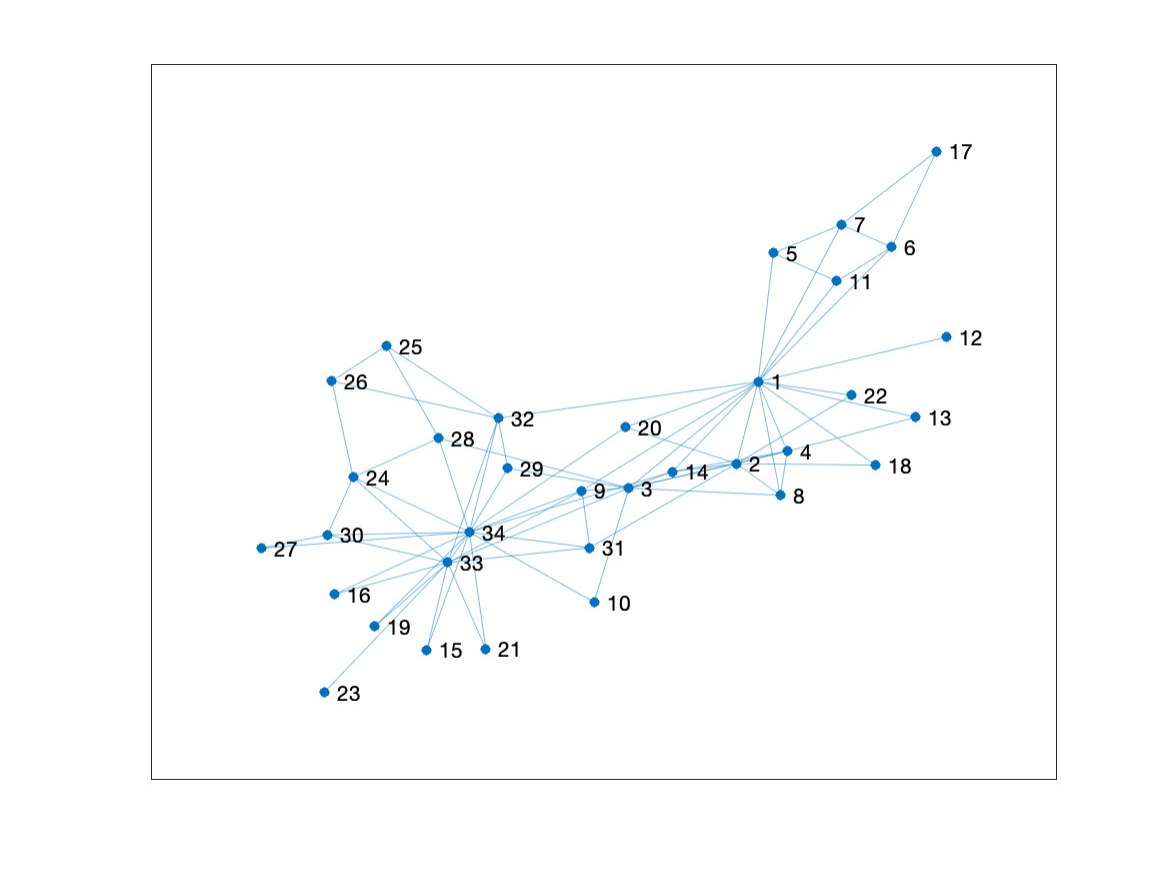} 
\includegraphics[scale = 0.4]{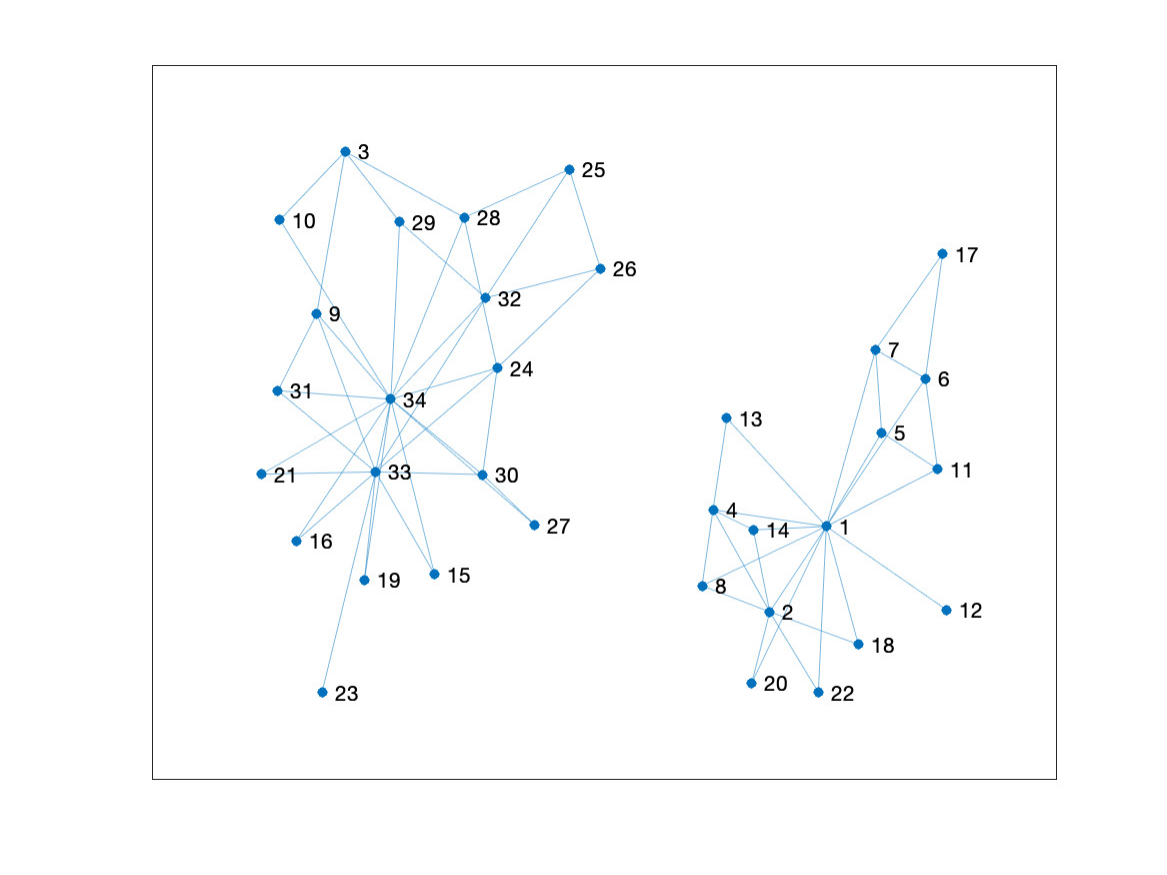} 
\end{tabular}
\caption{Example \ref{ex5}. The Zachary78 network (top left) and the 
disconnected graph determined by Algorithm 2 (top right), and the Zachary77 network 
(bottom left) and the disconnected graph determined by Algorithm 2 (bottom right).}
\label{fig3}
\end{figure}    

\begin{figure}
\centering
\begin{tabular}{cc}
\includegraphics[scale = 0.6]{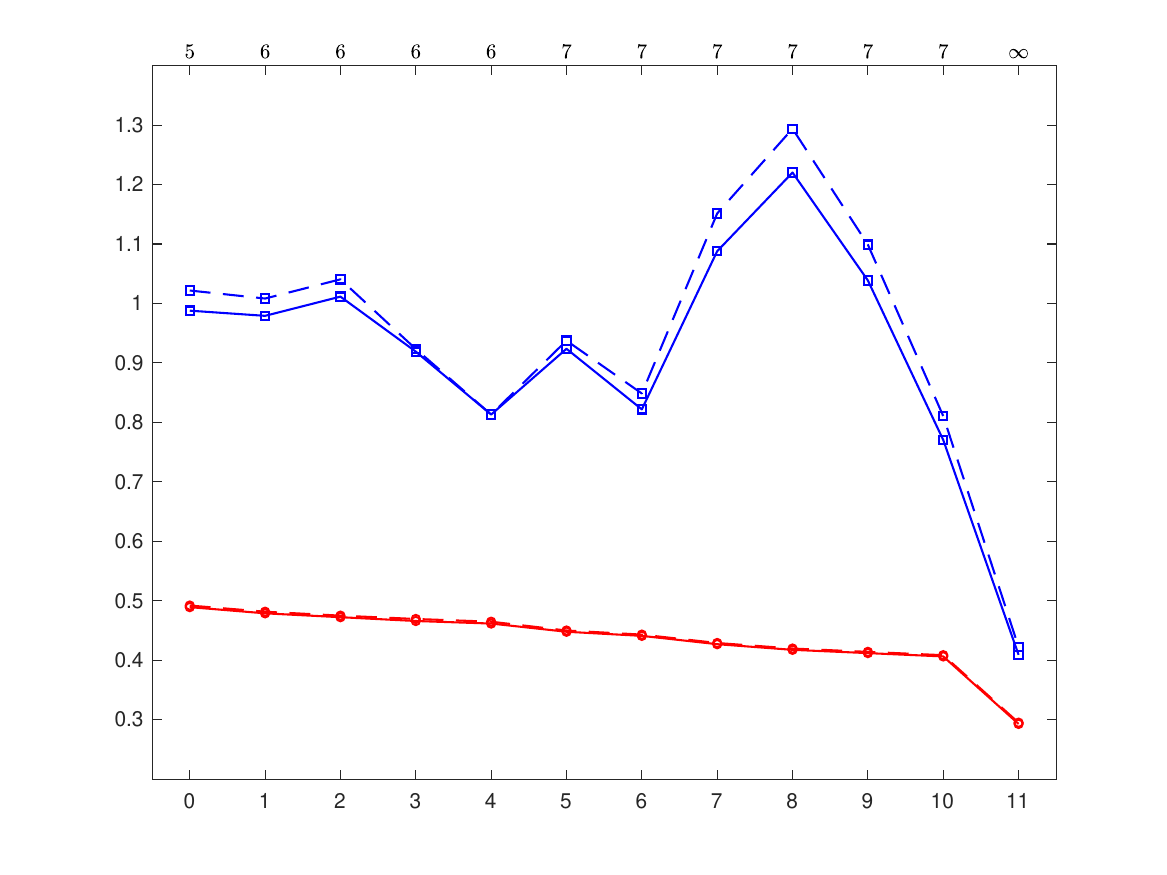} 
\end{tabular}
\caption{Example 5. Global efficiency (in red) and  global connectivity (in blue) are 
shown as functions of the number of steps of Algorithm 2 for both the Zachary78 network
(dashed lines) and the Zachary77 network (solid lines). Step $0$ refers to the original
networks. The numbers on the top horizontal axis display the graph diameters at each step. }
\label{fig4}
\end{figure}    

\end{example}

 {\section{Concluding remarks} \label{s6}
This paper is concerned with the flow of information via the shortest paths of a network.
These paths typically are the most important ones for transmitting information. To analyze
the effectiveness of the information flow via these paths, we introduced an edge 
centrality measure, as well as its analogue that only considers shortest paths and the 
number of these paths. These measures indicate whether an edge can be removed without
serious impact on the information flow. They refine the notion of edge betweenness 
centrality and its novel analogue that considers shortest paths only. Computed examples
illustrate that the novel measures introduced provide guidance in how to simplify an 
available network by removing certain edges.}

\section*{Acknowledgments}
 {The authors would like to thank a referee for comments that led to clarifications
of the presentation.}
Research by SN was partially supported by a grant from SAPIENZA Universit\`a di Roma, 
by INdAM-GNCS, and by the European Union - NextGenerationEU CUP F53D23002700006: PRIN 2022 
research project “Inverse Problems in the Imaging Sciences (IPIS)", grant n. 2022ANC8HL. 

\section*{Data availability statement} 
Data sharing not is applicable to this article as no new datasets were generated during the 
current study.

\section*{Conflict of interest} 
The authors declare that they have no conflict of interest. 

\section*{Funding}
See acknowledgments.

\section*{Authors' contributions}
All authors contributed equally to the paper.

\section*{Ethical Approval}
Not Applicable.


\begin{thebibliography}{99} 

\bibitem{AMDLCR}
M. Al Mugahwi, O. De la Cruz Cabrera, and L. Reichel, {\em Orthogonal expansion of network
functions}, Vietnam J. Math., 48 (2020), pp. 941--962.

\bibitem{ABCMP}
D. Altafini, D. A. Bini, V. Cutini, B. Meini, and F. Poloni, {\em An edge centrality 
measure based on the Kemeny constant}, SIAM J. Matrix Anal. Appl., 44 (2023), pp. 
648--669.

\bibitem{AB}
F. Arrigo and M. Benzi, {\em Edge modification criteria for enhancing the communicability 
of digraphs}, SIAM J. Matrix Anal. Appl., 37 (2016), pp. 443--468.

\bibitem{BFRR}
J. Baglama, C. Fenu, L. Reichel, and G. Rodriguez, {\em Analysis of directed networks via
partial singular value decomposition and Gauss quadrature}, Linear Algebra Appl., 456 
(2014), pp. 93--121 .

\bibitem{BBV}
A. Barrat, M. Barthelemy, and A. Vespignani, {\em Dynamical Processes on Complex Networks},
Cambridge University Press, Oxford, 2008.

\bibitem{Bor}
S. P. Borgatti, {\em Centrality and network flow}, Social Networks, 27 (2005), pp. 55--71.

\bibitem{Br}
D. Braess, {\em  \"Uber ein Paradoxon aus der Verkehrsplanung}, Unternehmensforschung, 12
(1968), pp. 258--268.

\bibitem{BC}
J. I. Brown and C. J. Colbourn, {\em Roots of the reliability polynomial}, SIAM J. 
Discrete Math., 5 (1992), pp. 571--585.

\bibitem{DJNR}
O. De la Cruz Cabrera, J. Jin, S. Noschese, and L. Reichel, {\em Communication in complex
networks}, Appl. Numer. Math., 172 (2022), pp. 186--205.

\bibitem{DMR}
O. De la Cruz Cabrera, M. Matar, and L. Reichel, {\em Analysis of directed networks via 
the matrix exponential}, J. Comput. Appl. Math., 355 (2019), pp. 182--192.

\bibitem{DMR2}
O. De la Cruz Cabrera, M. Matar, and L. Reichel, {\em Edge importance in a network via 
line graphs and the matrix exponential}, Numer. Algorithms, 83 (2020), pp. 807--832.

\bibitem{DMR3}
O. De la Cruz Cabrera, M. Matar, and L. Reichel, {\em Centrality measures for 
node-weighted networks via line graphs and the matrix exponential}, Numer. Algorithms, 88
(2021), pp. 583--614.

\bibitem{air500}
{\em Dynamic Connectome Lab - Data Sets}. 
\url{https://sites.google.com/view/dynamicconnectomelab}

\bibitem{estrada2011structure}
E. Estrada, {\em The Structure of Complex Networks: Theory and Applications}, Oxford 
University Press, Oxford, 2011.

\bibitem{EH}
E. Estrada and N. Hatano, {\em Communicability in complex networks}, Phys. Rev. E, 77 
(2008), Art.  036111.

\bibitem{estrada2005subgraph}
E. Estrada and J. A. Rodriguez-Velazquez, {\em Subgraph centrality in complex networks}, 
Phys. Rev. E, 71 (2005), Art. 056103.

\bibitem{FMRR}
C. Fenu, D. Martin, L. Reichel, and G. Rodriguez, {\em Network analysis via partial 
spectral factorization and Gauss quadrature}, SIAM J. Sci. Comput., 35 (2013), 
pp. A2046--A2068.

\bibitem{FRR}
C. Fenu, L. Reichel, and G. Rodriguez, {\em SoftNet: A package for the analysis of complex
networks}, Algorithms, 15 (2022), Art. 296.

\bibitem{Fr}
L. Freeman, {\em A set of measures of centrality based on betweenness}, Sociometry, 40 
(1977), pp. 35--41.

\bibitem{GN}
M. Girvan and M. E. J.  Newman, {\em Community structure in social and biological 
networks}, Proc. Natl. Acad. Sci. USA, 99 (2002), pp. 7821--7826.

\bibitem{HKYH}
P. Holme, B. J. Kim, C. N. Yoon, and S. K. Han, {\em Attack vulnerability of complex 
networks}, Phys. Rev. E, 65 (2002), Art. 056109.

\bibitem{ICON}
 {
\em Index of Complex Networks (ICON). 
\url{https://icon.colorado.edu/networks}
}

\bibitem{LPWS}
Y. Li, G-P. Jiang, M. Wu, and Y. Song, {\em  Methods to measure the network path 
connectivity}, Secur. Commun. Netw., 2020 (2020), pp. 1--10.

\bibitem{newman2010}
M. E. J. Newman, {\em Networks: An Introduction}, Oxford University Press, Oxford, 2010.

\bibitem{NR_na}
S. Noschese and  L. Reichel, {\em Network analysis with the aid of the path length 
matrix}, Numer. Algorithms, 95 (2024), pp. 451--470.

\bibitem{NR}
S. Noschese and L. Reichel, {\em Edge importance in complex networks}, Numer. Algorithms, 
99 (2025), pp. 377--410. 

\bibitem{SSMC}
 {
\em SuiteSparse Matrix Collections. 
\url{https://sparse.tamu.edu/Pajek}
}

\bibitem{PNTM}
D. Plav\v{s}i\'c, S. Nikoli\'c, N. Trinajsti\'c, and Z. Mihali\'c, {\em On the Harary index 
for the characterization of chemical graphs}, J. Math. Chem., 12 (1993), pp. 235--250.

\bibitem{Sc}
M. Schweitzer, {\em Sensitivity of matrix function based network communicability measures:
Computational methods and a priori bounds}, SIAM J. Matrix Anal. Appl., 44 (2023),
pp. 1321--1348.

\bibitem{Za}
 {W.W. Zachary, {\em An Information Flow Model for Conflict and Fission in Small Groups}, 
J. Anthro. Research 33  (1977), pp. 452--473.}


\end{thebibliography}
\end{document}